\documentclass[11pt]{article}
\usepackage{graphicx} % Required for inserting images
\usepackage{amsthm, amssymb, amsmath}
\usepackage{fullpage}
\usepackage{xcolor}

\usepackage{paralist}
\usepackage{hyperref}
\usepackage{cleveref}
\usepackage{comment}

\newcommand{\papertitle}{Identity Testing for Circuits with Exponentiation Gates}

\hypersetup{
    colorlinks=true,
    linkcolor=blue,
    citecolor=blue,
    filecolor=magenta,      
    urlcolor=cyan,
    pdftitle={\papertitle},
    pdfpagemode=FullScreen,
}

\usepackage{thmtools} 
\usepackage{thm-restate}

\title{\papertitle}

\author{Jiatu Li\thanks{Computer Science and Artificial Intelligence Laboratory, Massachusetts Institute of Technology, \href{mailto:jiatuli@mit.edu}{jiatuli@mit.edu}} \and Mengdi Wu\thanks{Computer Science Department, Carnegie Mellon University, \href{mailto:mengdiwu@andrew.cmu.edu}{mengdiwu@andrew.cmu.edu}}}
\date{\today}

\newcommand{\AExp}{\mathsf{AExp}}

\newcommand{\dom}{\mathsf{dom}}
\newcommand{\eqdef}{\triangleq}
\newcommand{\eps}{\varepsilon}

\newcommand{\F}{\mathbb{F}}
\newtheorem{theorem}{Theorem}
\newtheorem{lemma}[theorem]{Lemma}
\newtheorem{corollary}[theorem]{Corollary}
\newtheorem{proposition}[theorem]{Proposition}
\theoremstyle{definition}
\newtheorem{definition}{Definition}
\newtheorem{conjecture}{Conjecture}
\newtheorem{example}{Example}
\newtheorem{remark}{Remark}

\newcommand{\Q}{\mathbb{Q}}
\newcommand{\R}{\mathbb{R}}
\newcommand{\N}{\mathbb{N}}
\newcommand{\Z}{\mathbb{Z}}

\newcommand{\softmax}{\mathsf{softmax}}
\newcommand{\sigmoid}{\mathsf{sigmoid}}

\newcommand{\Attention}{\mathsf{Attention}}
\newcommand{\GLU}{\mathsf{GLU}}

\newcommand{\poly}{\mathsf{poly}}

\newcommand{\Quot}{\mathsf{Quot}}
\newcommand{\Inv}{\mathsf{Inv}}

\begin{document}

\maketitle

\begin{abstract}
    Motivated by practical applications in the design of optimization compilers for neural networks, we initiated the study of identity testing problems for arithmetic circuits augmented with \emph{exponentiation gates} that compute the real function $x\mapsto e^x$. These circuits compute real functions of form $P(\vec x)/P'(\vec x)$, where both $P(\vec x)$ and $P'(\vec x)$ are \emph{exponential polynomials}
    \[
    \sum_{i=1}^k f_i(\vec x)\cdot \exp\left(\frac{g_i(\vec x)}{h_i(\vec x)}\right),
    \]
    for polynomials $f_i(\vec x),g_i(\vec x)$, and $h_i(\vec x)$.
    
    We formalize a black-box query model over finite fields for this class of circuits, which is mathematical simple and reflects constraints faced by real-world neural network compilers. We proved that a simple and efficient randomized identity testing algorithm achieves perfect completeness and non-trivial soundness. Concurrent with our work, the algorithm has been implemented in the optimization compiler Mirage by Wu et al.~(OSDI 2025), demonstrating promising empirical performance in both efficiency and soundness error. Finally, we propose a number-theoretic conjecture under which our algorithm is sound with high probability. 
\end{abstract}

\section{Introduction}

Polynomial Identity Testing (PIT) is a central problem of theoretical computer science. Given a multi-variate polynomial $f:\F^n\to\F$ with certain properties, the goal of the problem is to verify whether $f$ is identically zero. 

Early results on identity testing of polynomials date back to DeMillo and Lipton \cite{DL78}, Zippel \cite{Zippel79}, and Schwartz \cite{Schwartz80} (see also \cite{Ore1921} for a special case on finite fields). It is proved that a black-box probabilistic testing, namely checking whether $f(\vec x)=0$ for a random element $\vec x\in\F^n$, works well for low-degree polynomials. This simple algorithm serves as a key step in many randomized algorithms (see, e.g., \cite{Tutte47,Lovasz79,AB03}) as well as probabilistic proof systems, including the celebrated interactive proof system for $\mathsf{PSPACE}$ \cite{Shamir92} and the PCP theorem \cite{AS98}. Subsequently, there has been a rich literature on identity testing of various types of polynomials, including sparse polynomials and polynomials computed by depth-$3$ algebraic circuits (see, e.g., \cite{Saxena09,Saxena14} and the references therein).

Perhaps surprisingly, the probabilistic identity testing algorithm for low-degree polynomials has recently been found useful in developing optimization compilers for deep neural networks. Loosely speaking, deep neural networks are represented as a high-level programming language, and the optimization compiler is designed to automatically identify redundancy in programs to improve the efficiency of the neural network. In a recent optimization compiler PET~\cite{PET21}, Wang et al.~observed that redundancy detection in partial neural network computation \emph{without} non-linear activation can be naturally modeled as black-box identity testing of low-degree polynomials\footnote{In the black-box setting, the algorithm is given oracle access to the polynomial rather than its description. For the optimization compilers of neural networks (e.g., \cite{PET21}), the neural networks are gigantic and are usually loaded into optimization devices such as GPUs or TPUs. Black-box identity testing algorithms are favorable, as \emph{evaluation} on (partial) neural networks is much more efficient than \emph{white-box access} to neural networks due to communication overhead and the design of the optimization devices.}. By applying the classical probabilistic testing algorithm \cite{DL78,Zippel79,Schwartz80}, PET achieves a significant speed-up over its benchmark. 

A natural question is, therefore, whether this approach can be extended to neural networks \emph{with} non-linear activation. In a subsequent work, Wu et al.~\cite{Mirage24} introduced a framework, Mirage, that could benefit from extending the redundancy detection algorithm in PET \cite{PET21} to neural network components with \emph{exponentiation} operators. This enables the modeling of non-linear activation functions such as $\softmax(\cdot)$. In the language of circuit complexity, the program representation in Mirage can be modeled as integer-coefficient arithmetic circuits with ``exponentiation gates'', where there is at most one exponentiation gate in each path from output gates to input gates (see \Cref{sec:def_circuit_model}). The goal of Mirage is to ``efficiently'' check whether a circuit $C:\R^n\to\R$ of such form is identical to $0$ on $\mathbb{R}$ in a black-box query model that is practically satisfactory.

Compared to polynomials and standard arithmetic circuits computing polynomials, the behavior of circuits with exponentiation gates is not well studied. Integer-coefficient polynomials can be naturally evaluated over $\F_p$ with modular arithmetic. However, it is not a-prior clear what it the correct way to define the evaluation of circuits with exponentiation gates over finite fields, as there is no standard analogy of the exponentiation function in finite fields. This makes it hard to even define a suitable ``black-box query model'' for the identity testing problem that is easy to analyze and captures the real-world constraints. Moreover, it is unclear whether the standard randomized algorithm \cite{DL78,Schwartz80,Zippel79} works after we define the evaluation of such circuits. 

\paragraph{Our Contribution.} Motivated by the practical applications, we initiate a theoretical analysis of the identity testing of \emph{circuit with exponentiation gates}. 
\begin{itemize}
\item We introduce a natural circuit model that formalizes circuits with exponentiation gates $x\mapsto \exp(x)$, which suffices to captures components of modern neural networks such as Attention \cite{attention} with certain non-linear activation functions. The circuit model generalizes the standard algebraic circuit model by allowing exponentiation gates. 
\item We prove that in the \emph{idealized model} where the algorithm could query the circuit on any \emph{real number}, a simple probabilistic testing is correct with high probability. 
\item We introduce an \emph{algebraic query model} that captures the real-world constraints in optimization compilers for deep neural networks such as Mirage \cite{Mirage24}, and design a simple randomized algorithm that is perfectly complete and \emph{non-trivially} sound. Concurrent with our work, this algorithm is implemented in Mirage as one of its two redundancy detection algorithms. We also introduce a conjecture about sparse polynomials in finite field, under which our randomized algorithm in algebraic query model is sound \emph{with high probability}.  
\end{itemize}

\subsection{Main Results: A Simplified Setting}

Before formally describing our circuit and query models in \Cref{sec: models}, we try to describe a simplified version of the main mathematical problem. Let $k\in\N$, $f_i,g_i,h_i$ be $n$-variate degree-$d$ polynomials with integer coefficients in $[-w,w]$, and $P:\R^n\to\R\cup\{\bot\}$ be a partial real function defined as 
\[
P(\vec x) \eqdef \sum_{i=1}^{k} f_i(\vec x) \exp\left(\frac{g_i(\vec x)}{h_i(\vec x)}\right),
\]
where $\exp(x)\eqdef e^x$. This function is similar to the \emph{exponential polynomials} that are used in, e.g., transcendental number theory \cite[Chapter 12]{Baker-book}. For simplicity, we will call $P(\vec x)$ a degree-$d$ exponential polynomial. The number of terms $k$ in $P(\vec x)$ is said to be the \emph{width} of $P(\vec x)$. We call $\{f_i(\vec x)\}_{i \in [k]}$ the coefficient polynomials of $P(\vec x)$, and $\{g_i(\vec x)/h_i(\vec x)\}_{i \in [k]}$ the exponent fractions of it.  

The algorithmic problem we are interested in is to check whether $P(\vec x)$ is identically zero on $\R^n$ with queries of one of the following two forms: 
\begin{itemize}
\item (\emph{Real}): Given $\vec x\in\Q^n$, the oracle outputs $\{0,1,\bot\}$ indicating that $P(\vec x) = 0$, $P(\vec x)\in\R\setminus\{0\}$, and $P(\vec x)=\bot$ (i.e.~undefined due to division-by-zero), respectively. 
\item (\emph{Algebraic}): Let $p,q$ be prime numbers such that $q\mid p-1$, $G=\langle g\rangle$ be the unique multiplicative subgroup of $\F_p^*$ of order $q$. Given $\vec u\in\F_p^n, \vec v\in\F_q^n$, and $a\in G$, the oracle outputs 
\[
P_a(\vec u,\vec v)\eqdef \left(\sum_{i=1}^k f_i(\vec u) \cdot a^{g_i(\vec v)\cdot (h_i(\vec v))^{-1}\bmod q}\right) \bmod p,  
\]
where $(\cdot)^{-1}$ on the exponent denotes the multiplicative inverse in $\F_q$. The oracle returns $\bot$ if $h_i(\vec v) \bmod q = 0$ for any $i\in[k]$. 
\end{itemize}

The former type of query is an idealized model. In practical applications (e.g.~\cite{PET21,Mirage24}), there is no efficient algorithm to implement the oracle that checks whether $P(\vec x) = 0$ precisely given $\vec x$, and the approximation of real numbers using floating-point numbers is also unsatisfactory in practice (see \Cref{sec: models} for more discussion). The latter type of query is more realistic; indeed, it is implemented in the optimization compiler Mirage \cite{Mirage24} where the explicit representation of $P(\vec x)$ is a (partial) neural network.\footnote{The special case of the algebraic query model without exponentiation functions is implemented in \cite{PET21}.} 

\paragraph{Identity Testing Algorithms in the Simplified Setting.} Now we describe our identity testing algorithms in the simplified setting. First, we show that a simple randomized algorithm with \emph{one} query works in the real query model. The algorithm is to randomly sample $\vec x\in[B]^n$ for $B = 20\cdot d\cdot k^2$ and accepts if $P(\vec x)\in\{0,\bot\}$. It is clear that the algorithm is perfectly complete, and the soundness is given by the following theorem: 

\begin{theorem}[\Cref{thm: soundness}, simplified]\label{thm: real model informal}
Let $P(\vec x)$ be an $n$-variate degree-$d$ exponential polynomial of width $k$ that is not identically zero on its domain, and $S\subseteq \Q$ be any finite non-empty set. For $\vec x\in S^n$ sampled uniformly at random, $\Pr[P(\vec x) \in \{0,\bot\}]\le 8dk^2/|S|$. 
\end{theorem}

As the second result, we analyze the following simple randomized identity testing algorithm in the algebraic query model: Let $p,q,G$ be defined as above, the algorithm randomly samples $\vec u\in\F_p^n$, $\vec v\in \F_q^n$, and $a\in G$, and accepts if $P_a(\vec u,\vec v) \in \{0,\bot\}$. The following theorem formalizes the completeness and soundness of the algorithm: 

\begin{theorem}[\Cref{thm: algebraic model}, simplified]\label{thm: correctness algebraic informal}
Let $P(\vec x)$ be an $n$-variate degree-$d$ exponential polynomial of width $k$. Let $p,q$ be prime numbers, $q\mid p-1$, and $G=\langle g\rangle$ be the unique order-$q$ multiplicative subgroup of $\F_p^*$. Suppose that all integer coefficients in $P(\vec x)$ are within $[-w,w]$, $q > 2w$, then: 
\begin{compactitem}
\item (Completeness). If $P(\vec x)$ is identically zero on its domain as a partial real function, for any $a\in G$, $\vec u\in\F_p^n$, and $\vec v\in\F_q^n$, $P_a(\vec u,\vec v)\in\{0,\bot\}$. 
\item (Soundness). If $P(\vec x)$ is not identically zero on its domain as a partial real function, $p,q \ge 2kw$, then for uniformly random $a\in G$, $\vec u\in\F_p^n$, and $\vec v\in\F_q^n$, the probability that $P_a(\vec u,\vec v) \in \{0,\bot\}$ is at most $q^{-\frac{1}{k-1}}+O(dk^2/q)$. 
\end{compactitem}
\end{theorem}

Note that for sufficiently large $k,q\in\N$ such that $\ln q \le k-1$, we have 
\[
1-\frac{2\ln q}{k-1}\le q^{-\frac{1}{k-1}} \le  1-\frac{\ln q}{2(k-1)}.
\]
Thus in a typical setting that $dk^3\le q\le k^{O(1)}$, the soundness error is $1-\Theta(\log k/(k-1))$.\footnote{The error probability is constant if $q$ is exponential in $k$. However, this setting is not meaningful: The width $k$ is usually comparable to the input length, and thus the arithmetic operations over $q$ would be extremely inefficient.} By parallel repetition of the randomized algorithm for $O(k\log k/\log q)$ times, we can boost the error probability to $1/k^{O(1)}$. This leads to an efficient randomized identity testing algorithm when $k$ is relatively small and the evaluation is much more efficient than obtaining the description of $P(\vec x)$, which, in particular, captures the real-world constraints for the optimization compilers such as Mirage \cite{Mirage24}. 

\subsection{Connection to the Circuit Models} \label{sec: intro circuits}

We stress that the abstraction in previous subsection is mathematically clean but omits crucial details in modeling the real-world problem. Specifically: 
\begin{itemize}
    \item It is not immediately clear how the identity testing results for exponential polynomials apply to our motivating real-world application, namely identity testing of neural networks formalized by ``circuits with exponentiation gates'' \cite{PET21,Mirage24}. 
    \item It is also not clear how the function $P_{a}(\vec u,\vec v)$ in \Cref{thm: algebraic model} relates to the circuits, and why it can be computed efficiently when $P(\vec x)$ is explicitly given by a ``circuit with exponentiation gates''. This is important as in real-world applications, see \cite{Mirage24}, the function $P_{a}(\vec u,\vec v)$ is implemented on specialized devices such as GPU or TPU, which are optimized for specific computation patterns. Moreover, it is unclear where the restrictions in the algebraic query model (e.g., $a$ must be in $G$) come from. 
\end{itemize}

\paragraph{Circuit Model.} We briefly explain our circuit model (called $\AExp^1$ circuits) and refer readers to \Cref{sec: models,sec: exp polynomials} for more details. We work with (arithmetic) circuits with integer coefficients, unbounded fan-in addition and multiplication gates, fan-in two division gates, and fan-in one exponentiation gates. In addition, we impose the restriction that on each path from input variables to output gates, there is at most \emph{one} exponentiation gates --- the circuit cannot compute double exponential $x\mapsto e^{e^x}$ by composing exponentiation gates. 

The evaluation of circuits over real numbers is defined by a standard gate-by-gate evaluation algorithm, where the exponentiation gate is interpreted as the function $x\mapsto e^x$. However, it is unclear how to define the evaluation of such circuits over finite fields, as there is no standard interpretation of the exponential function. 

Let $p$ be a prime number. To define the evaluation of an $\AExp^1$ circuit $C$ over $\F_p$, a natural idea is to replace the exponential function by $x\mapsto a^x$ for some element $a\in \F_p$. This is not ideal as such an interpretation is \emph{inconsistent} with the evaluation over $\R$: It could be the case that 
\[ 
a^{x_1+x_2\bmod p}\not\equiv a^{x_1} a^{x_2}\pmod p
\] 
for $x_1,x_2\in\F_p$, while $e^{x_1+x_2} = e^{x_1}e^{x_2}$ for any $x_1,x_2\in\R$. As a result, this definition cannot be used for the identity testing of $\AExp^1$ circuits. 

To address the issue, we exploit the algebraic structure of finite fields by using different moduli over and under the exponents. Let $p,q$ be prime numbers such that $q\mid p-1$, $G\subseteq\F_p^*$ be the unique order-$q$ multiplicative subgroup, and $a\in G$. It follows that 
\[
a^{x_1+x_2\bmod q} \equiv a^{x_1}\cdot a^{x_2}\pmod  p
\]
for any $x_1,x_2\in\F_q$. If we use $q$ as the modulus ``over the exponent'' and $p$ as the modulus ``under the exponent'', the function $x\mapsto a^x$ will be consistent with the arithmetic law $e^{x_1+x_2} = e^{x_1}\cdot e^{x_2}$. Following the intuition, we define the evaluation of $C$ with respect to $(p,q,a)$ using a gate-by-gate evaluation algorithm such that: 
\begin{compactitem}
\item Each input variable or wire carries a pair $(u,v)\in(\F_p\cup\{\bot\})\times(\F_q\cup\{\bot\})$. 
\item Addition, multiplication, and division gates are implemented coordinate-wisely. 
\item The exponentiation gate is implemented by $(u,v)\mapsto (a^{v}\bmod p, \bot)$. 
\end{compactitem}
This evaluation algorithm can be implemented efficiently given the description of the circuit $C$. We refer readers to \Cref{sec: models} for more details. 

\paragraph{A Structural Lemma.} Let $\vec u=(u_1,\dots,u_n)\in\F_p^n$, $\vec v=(v_1,\dots,v_n)\in\F_q^n$, we use $C_{a}(\vec u,\vec v)$ to denote the output of the evaluation algorithm where the $i$-th input variable is assigned to $(u_i,v_i)$. Similar to the standard results that algebraic circuits compute polynomials \cite[Section 16]{AB09}, we can then prove that $\AExp^1$ circuits can be converted to an equivalent \emph{fraction of exponential polynomials}: 

\begin{lemma}[\Cref{lmm: normal form}, informal]\label{lmm: structural informal}
    For every $n$-input $\AExp^1$ circuit $C$, there are $k,k'\in\N$ and $n$-variate integer-coefficient exponential polynomials $P(\vec x)$ and $P'(\vec x)$ such that the following holds: 
    \begin{compactitem}
    \item For each $\vec u\in\vec\R$, $C(\vec u) = P(\vec u)/P'(\vec u)$. 
    \item Let $p,q$ be prime numbers such that $q\mid p-1$, $G_{p,q}$ be the multiplicative subgroup of $\F_p^*$ of order $q$, and $a\in G_{p,q}$. Then for every $(\vec u,\vec v)\in\vec\F_p\times \vec \F_q$, $C_a(\vec u,\vec v)\equiv P_a(\vec u,\vec v)\cdot(P'_a(\vec u,\vec v))^{-1}\pmod p$, where $(a)^{-1}$ denotes the multiplicative inverse of $a$ modulo $p$. 
    \end{compactitem}
\end{lemma}

We stress that the structural lemma \emph{does not} provide a polynomial upper bound on the degree, width, or integer weight of the integer-coefficient exponential polynomials $P,P'$. Nevertheless, it can be verified that the exponential polynomials $P,P'$ for $\AExp^1$ circuits from real-world neural network applications, such as $\softmax(\cdot)$ and Attention with softmax activations \cite{attention}, tend to have relatively small degree, width, and weight; interested readers are referred to \Cref{sec: real world NN} for more discussion, and \cite{Mirage24} for experimental results.  

\paragraph{The Main Theorem.} Let $\dom_\R(C)$ be the domain of the circuit $C$, i.e., the set of $\vec x\in\R^n$ such that $C(\vec x)\ne \bot$. An $\AExp^1$ circuit $C$ is said to have degree $d$, width $k$, and weight $w$ if there are integer-coefficient degree-$d$ width-$k$ exponential polynomials $P(\vec x)$ and $P'(\vec x)$ with all integer coefficients within $[-w,w]$ that satisfy \Cref{lmm: structural informal}. 

By generalizing \Cref{thm: real model informal,thm: correctness algebraic informal} to fractions of exponential polynomials and combining them with \Cref{lmm: structural informal}, we can obtain the final result: 

\begin{restatable}{theorem}{MainTheorem}\label{thm: algebraic model}
    Let $C:\R^n\to\R$ be an $\AExp^1$ circuit of width $k$, degree $d$ and weight $w$, $p$ and $q$ be prime numbers such that $q\mid p-1$ and $q > 2(kw)^2$. Let $G_{p, q}$ be the unique multiplicative subgroup of $\F_p^*$ of order $q$. The following hold:
    \begin{compactitem}
    \item (Completeness). If $C$ is identically zero on $\dom_\R(C)$, then for every $a \in G_{p, q}$, $C_{a}(\vec u,\vec v)\in \{0,\bot\}$ for every $(\vec u,\vec v)\in\F_p^n\times \F_q^n$.
    \item (Soundness). If $C$ is not identically zero on $\dom_\R(C)$, then for uniformly random $\vec u\gets \F_p^n, \vec v\gets\F_q^n,a\gets G_{p,q}$, $\Pr[C_{a}(\vec u,\vec v)\notin \{0, \bot\}]\ge 1 - 8dk^4\cdot q^{-1} - q^{-1/(k^2-1)}$.
    \end{compactitem}
\end{restatable}
  
\subsection{Technical Overview}

We briefly explain the proof of our simplified technical results \Cref{thm: real model informal} and \Cref{thm: correctness algebraic informal} in the simplified setting, as well as the main theorem (see \Cref{thm: algebraic model}) that generalizes the results to the circuit setting.

\paragraph{Real Queries.} The idea behind \Cref{thm: real model informal} is very simple: We manage to combine the Schwartz-Zippel lemma and Lindemann-Weierstrass theorem, i.e., $e$ is \emph{transcendental} (see \Cref{thm: LW}). 

To avoid technical subtlety, we assume that $P(\vec x)$ is a degree-$d$ exponential polynomial where the exponent fractions $g_i(\vec x)/h_i(\vec x)$ are pairwisely distinct; that is, for every pair of indices $i,j\in[k]$, $i\ne j$, $g_i(\vec x)h_j(\vec x) - g_j(\vec x)h_i(\vec x)$ is not a zero polynomial. Such an exponential polynomial is said to be \emph{condensed}, and the general case can be reduced to the condensed case by considering another exponential polynomial $\hat P$ that ``groups'' coefficients based on the exponent fraction; see \Cref{sec: condensation}. Moreover, we assume without loss of generality that $f_i\ne 0$ and $h_i\ne 0$ for each $i\in[k]$. 

Let $S\subseteq\Q$ be a non-empty finite set, we know by the Schwartz-Zippel lemma that for any $i\in[k]$, $j\in[k]$, $i\ne j$, each of the following non-zero polynomials of degree at least $2d$ 
\[
f_i(\vec x), \quad h_i(\vec x), \quad g_i(\vec x) h_j(\vec x) - g_j(\vec x) h_i(\vec x) 
\]
evaluates to zero with probability at most $2d / |S|$ for a uniformly random $\vec x\in S^n$. By the union bound, we can further prove that with probability at least $1-3dk^2/|S|$, none of the polynomials evaluate to $0$. For each of such $\vec x$, we can see that for the rational numbers $\alpha_i \eqdef g_i(\vec x)/h_i(\vec x)$ and $\beta_i\eqdef f_i(\vec x)$, we have 
\[
P(\vec x) = \beta_1 e^{\alpha_1} + \beta_2 e^{\alpha_2} + \dots + \beta_k e^{\alpha_k},
\]
which must be non-zero as $e$ is transcendental. 

\paragraph{Connection to Algebraic Queries.} Unwinding the proof of \Cref{thm: real model informal}, we could actually give a necessary and sufficient condition for an exponential polynomial with integer weights to be identically zero over real evaluations. Let $P(\vec x)$ be a condensed degree-$d$ exponential polynomial 
\[
P(\vec x) = \sum_{i=1}^k f_i(\vec x) \exp\left(\frac{g_i(\vec x)}{h_i(\vec x)}\right)
\]
with integer coefficients, then $P(\vec x)$ is identically zero on its domain (over real evaluation) if and only if for every $i\in[k]$, $f_i = 0$. Subsequently, if $p,q$ are chosen to be larger than the integer coefficients, the identity testing of a condensed exponential function $P(\vec x)$ is equivalent to testing whether $f_i = 0\pmod p$ for every $i\in[k]$. 

With the characterization above, the completeness of \Cref{thm: correctness algebraic informal} is relatively simple, so we will focus on the soundness property. 

We try to mimic the strategy in the proof of \Cref{thm: real model informal}. Suppose that $P(\vec x)$ is not identically zero on real evaluations, we know by the discussion above that when $p$ are sufficiently large, there is at least one $f_i\ne 0\pmod p$ for any $i\in[k]$. Moreover, if $q$ is sufficiently large, we know that for any $i\ne j$, $g_ih_j - g_jh_i\ne 0\pmod q$. Assume without loss of generality that $f_i\ne 0$ for all $i\in[k]$. For $\vec u\in\F_p^n$ and $\vec v\in\F_q^n$ sampled uniformly at random, we still know by the Schwartz-Zippel lemma and the union bound that with probability at least $1-3d k^2 / q$, none of 
\[
f_i(\vec x)\bmod p, \quad h_i(\vec x)\bmod q, \quad g_i(\vec x) h_j(\vec x) - g_j(\vec x) h_i(\vec x) \bmod q
\]
evaluates to zero for $i,j\in[k]$, $i\ne j$. We can then define $\alpha_i \eqdef g_i(\vec x)/h_i(\vec x)\bmod q$ and $\beta_i\eqdef f_i(\vec x)\bmod p$ for $i\in[k]$ such that 
\[
P_a(\vec u,\vec v) = \beta_1 a^{\alpha_1} + \dots \beta_k a^{\alpha_k} \bmod p, 
\]
where $\beta_1,\dots,\beta_k$ are non-zero and $\alpha_1,\dots,\alpha_k$ are pairwisely distinct. 

\paragraph{A Weak Descartes' Rule in Finite Fields.} It then suffices to prove that for any non-zero $\beta_1,\dots,\beta_k\in\F_p$ and distinct $\alpha_1,\dots,\alpha_k\in\F_q$, $k\ge 1$, if we sample $a\in G$ uniformly at random, the probability that $\beta_1 a^{\alpha_1} + \dots \beta_k a^{\alpha_k} \bmod p =0$ is at most $q^{-1/(k-1)}$. We note that this can be viewed as a generalization of the \emph{Descartes' Rule} (see also \cite{Kelley16}), which asserts that any univariate polynomial with $k$ monomials has at most $2k$ roots. 

We provide an elementary proof of a weaker result: The probability that $\beta_1 a^{\alpha_1} + \dots \beta_k a^{\alpha_k} \bmod p =0$ is at most $1-1/k$. \Cref{thm: correctness algebraic informal} is proved using a lemma from \cite{Kelley16}, which employs a similar but slightly more complicated argument. 

Recall that the order-$q$ multiplicative subgroup $G\subseteq\F_p^*$ is a cyclic group. For any fixed $i\in\F_q$, a random element $a\in G$ can be viewed as generated from randomly sampling $g\in G$ and outputting $g^i$. Our idea is to choose a good $i\in\F_q$ such that at least a $1/k$ fraction of $g$ satisfies that $\beta_1 g^{i\cdot \alpha_1} + \dots \beta_k g^{i\cdot \alpha_k} \bmod p \ne 0$. 

We say that an index $i\in\F_q$ is \emph{good} for $\alpha\in\F_q$ if $i\cdot\alpha \bmod q \le (1-1/k) \cdot q$, and is said to be \emph{good} if for every $j\in[k]$, $i$ is good for $\alpha_i$. It turns out that if there is a good $i\in\F_q$, 
\[
\beta_1 g^{i\cdot \alpha_1} + \dots + \beta_k g^{i\cdot \alpha_k}
\]
can be viewed as a polynomial in $\F_p[g]$ of degree at most $(1-1/k)\cdot q$; in that case, there are at most $(1-1/k)\cdot q$ roots in $\F_p$ (and also in $G$). 

Therefore, it remains to prove the existence of a good $i$. Since $i\mapsto i\cdot\alpha$ is a bijection in $\F_q$, we know that for any fixed $\alpha$, the probability that a random $i\in\F_q$ is good for $\alpha$ is at least $1-1/k$. By the union bound, we know that a random $i\in\F_q$ is good for $\alpha_1,\dots,\alpha_k$ with non-zero probability. This implies that there must be a good $i\in\F_q$, which completes the proof.  

\begin{remark}[Related results]
The analogy of Descartes' Rule over finite fields has been studied prior to our work. Motivated by understanding the security of the Diffie-Hellman cryptosystem, Canettie et al.~\cite{Canetti00} proved an upper bound on the number of roots of sparse univariate polynomials. This upper bound was later improved by Kelley \cite{Kelley16}. We use the techniques from \cite{Canetti00,Kelley16} and obtain similar upper bounds. Note that \cite{Canetti00,Kelley16} considers the number of roots in $\F_p$, while we consider the number of roots in a multiplicative subgroup $G\subseteq\F_P^*$; as a result, our upper bound is cleaner and easier to prove. 
\end{remark}

\paragraph{Generalization to the Circuit Setting.} There are a few technical issues to obtain the main theorem (see \Cref{thm: algebraic model}). 

First, the proof overview above assumes that the exponential polynomial is \emph{condensed}, i.e., its exponent fractions are pairwisely distinct. For exponential polynomials that are not condensed, we need to first \emph{condense} the polynomial by merging terms with equivalent exponent fractions. For instance, the following exponential polynomial 
\[
P(\vec x) = \exp\left(\frac{x^2-3x}{x-3}\right) + \exp\left(\frac{x^2-2x}{x-2}\right)
\]
may be condensed to 
\[
\hat P(\vec x) = 2\exp\left(\frac{x^2-3x}{x-3}\right). 
\]
In general, such condensation procedure results into another exponential polynomial $\hat P$ that has larger domain and agree with $P$ on the domain of $P$ (see \Cref{prop: condensation}). We need to bridge the gap between $P$ and $\hat P$ with standard probability analysis. 

Second, as shown in \Cref{sec: intro circuits}, $\AExp^1$ circuits are converted to fractions $P/P'$ of exponential polynomials rather than exponential polynomials. This requires a careful (but straightforward) adaption of the techniques above. In particular, in the soundness analysis, we use the observation that $P/P'$ is identically zero on its domain (over $\R$) if and only if the exponential polynomial $P\cdot P'$ is identically zero on its domain (over $\R$). This leads to a quadratic overhead (in $k$) in the soundness error of \Cref{thm: algebraic model} compared to \Cref{thm: correctness algebraic informal}, as $P\cdot P'$ may have width up to $k^2$ when both $P$ and $P'$ are of width $k$.  

\paragraph{Organization of the Paper.} We review basic definitions and classical results that we will need in \Cref{sec: preliminaries}. In \Cref{sec: models}, we formally describe our circuit model as well as the query models we considered in our paper --- the idealized real query model and the algebraic query model; we also briefly explain why the algebraic query model is a better abstraction in the application of optimization compilers for neural networks. In \Cref{sec: exp polynomials}, we define exponential polynomials, discuss its basic properties, and prove the structural lemma that converts circuits to fractions of exponential polynomials. In \Cref{sec: algorithms}, we prove the correctness of our probabilistic testing algorithm in the real and algebraic query models.

\subsection{Discussion and Open Problems}

The most interesting open problem is to improve the soundness error of \Cref{thm: correctness algebraic informal} and \Cref{thm: algebraic model}. We conjecture that the actual soundness error is $1-\Omega(1)$. In particular, we propose the following number-theoretic conjecture under which the soundness error is indeed $1-\Omega(1)$:

\begin{conjecture}[Strong Descartes' Rule Conjecture over Finite Fields]
    There are constants $\eps,\delta<1$ such that the following holds. Let $p$, $q$ be sufficiently large prime numbers such that $q\mid p-1$, $g\in\F_q$ be an element of order $g$, and $G\eqdef\langle g \rangle$ be the unique multiplicative subgroup of $\F_p^*$ of order $q$. Let $k\le q^\delta$, $\alpha_1,\dots,\alpha_k\in \F_q$ be distinct, and $\beta_1,\dots,\beta_k\in\F_p$. Then the univariate polynomial
    \[
    f(z)\eqdef \beta_1 z^{\alpha_1} + \beta_2 z^{\alpha_2} + \dots + \beta_k z^{\alpha_k}
    \]
    has at most $\eps \cdot q$ roots in $G$.
\end{conjecture}

This conjecture can be interpreted as a property of the isomorphism mapping 
\[ 
I:[f] \mapsto (f(1), f(g), \dots, f(g^{q-1}))
\] 
from $\F_p[x]/(x^q-1)$ onto $\F_p^q$ (see \Cref{cor: isomorphism}). It states that there is no non-zero polynomial $[f]\in\F_p[x]/(x^q-1)$ such that both $f$ and $I([f])$ are very sparse. This seems to be an analogy of the ``uncertainty principle'' in Fourier analysis (see, e.g., \cite[Exercise 3.15]{Boolean-function}). We also note that the numerical experiments (see, e.g., \cite{Kelley16,CGRW17}) suggest that it is hard to construct sparse polynomials over $\F_p$ with many roots for prime $p$. 

\paragraph{Complexity-theoretic Perspectives.} It is also interesting to consider whether we can design better identity testing algorithms for some restricted classes of $\AExp^1$ circuits (e.g.~of small constant depth). This would potentially lead to real-world applications, as the structure of circuits from neural networks is relatively simple. To start with, one may consider adapting existing techniques for identity testing of algebraic circuits (see, e.g., \cite{Saxena09,Saxena14}) to $\AExp^1$ circuits.

On the other hand, it is interesting to consider whether there are conditional or unconditional \emph{lower bounds} for identity testing of $\AExp^1$ circuits in black-box models, such as the algebraic query models that we introduced in \Cref{sec: intro circuits}. 

\paragraph{Other Modeling of the Problem.} We note the the algebraic query model in \Cref{sec: intro circuits} is not necessarily the only reasonable formalization of the real world problem. Recall that in optimization compilers for neural networks \cite{PET21,Mirage24}, the neural network (modeled by an $\AExp^1$ circuit) is loaded into specialized devices (such as GPU or TPU) that are optimized for specific computation patterns and have high communication overhead with the CPU. Algorithms in other black-box query models are potentially useful for real-world applications if: 
\begin{compactenum} 
\item the queries can be efficiently implemented on those specialized devices; and
\item the communication overhead is small.
\end{compactenum} 
Note that it may make it possible to design better identity testing algorithms if we work with another black-box query model that has better mathematical properties. 

\section*{Acknowledgement}

We are grateful to Seyoon Ragavan for mentioning the connection of our results to \cite{Kelley16} and related works. We also thank Ryan O'Donnell and Ryan Williams for discussions about the identity testing algorithm and its analysis, and Zhihao Jia and Oded Padon for discussions about modeling the real-world problem. Jiatu Li received support from the National Science Foundation under Grant CCF-2127597. Mengdi Wu received support from NSF awards CNS-2147909, CNS-2211882, and CNS-2239351 and research awards from Amazon, Cisco, Google, Meta, NVIDIA, Oracle, Qualcomm, and Samsung.

\section{Preliminaries}
\label{sec: preliminaries}

\subsection{Abstract Algebra}

We assume basic familiarity to elementary ring and field theory (see, e.g., \cite{Abstract04}). We will use the standard notation: $R[x_1,\dots,x_m]$ denotes the ring of $m$-variate polynomials with coefficients in $R$; for any ideal $I$ in $R$, $R/I$ denotes the quotient ring of $R$ modulo $I$; $\Quot(R)$ denotes the quotient field (i.e.~the field of fraction) extending an integral domain $R$. For any prime $p$ and $u\in\F_p\setminus\{0\}$, we use $\Inv_p(u)$ to denote the multiplicative inverse of $u$ modulo $p$.

We will need a representation theorem for $\F_p[x]/(x^q-1)$ from the Chinese Remainder Theorem. 

\begin{theorem}[Chinese Remainder Theorem {\cite[Section 7.6]{Abstract04}}]
Let $R$ be a commutative ring with identity $1\ne 0$, $A_0,\dots,A_{k-1}$ be ideals in $R$. The map $R\to R/A_0\times \dots R/A_{k-1}$ defined by $r\mapsto (r+A_0,\dots,r+A_{k-1})$ is a ring homomorphism with kernel $A_0\cap \dots\cap A_{k-1}$. 
\end{theorem}

\begin{corollary}\label{cor: isomorphism}
Let $p,q$ be prime numbers such that $q\mid p-1$. Let $G = \langle g\rangle$ be the unique multiplicative subgroup of $\F_p^*$ of order $q$. Then $\F_p[x]/(x^q-1)\cong \F^q_p$, and the map $E : \F_p[x]/(x^q-1) \to \F^q_p$ defined by $f\mapsto (f(1), f(g), \dots, f(g^{q-1}))$ is an isomorphism.   
\end{corollary}
\begin{proof}[Proof Sketch]
Let $A_i = \langle x-g^i\rangle$ for $i\in \{0,1,\dots,q-1\}$ be an ideal in $\F_p[x]$. Clearly, we have that $\F_p[x]/A_i\cong \F_p$ by the isomorphism $f+A_i\mapsto f(g^i)$. By the Chinese Remainder Theorem, we have that the map $R:\F_p[x]\to \F_p[x]/A_0\times \dots\times \F_p[x]/A_{q-1}$ defined by $f\mapsto (f+A_0,\dots,f+A_{q-1})$ is a homomorphism with kernel $A_0\cap \dots\cap A_{k-1} = \F_p/(x^q - 1)$. Therefore, $R$ is an isomorphism from $\F_p[x]/(x^q-1)$ onto $\F_p[x]/A_0\times \dots\times \F_p[x]/A_{q-1}$. The corollary then follows by composing $R$ and the isomorphism from $\F_p[x]/A_i$ onto $\F_p$. 
\end{proof}

Recall that an \emph{integral domain} is a non-zero commutative ring where the multiplication of two non-zero elements is non-zero. We will need the following result for polynomials. 

\begin{lemma}[{\cite[Proposition 1 of Section 9.1]{Abstract04}}]\label{lmm: polynomial integral domain}
    For any integral domain $R$, the ring of $R$-coefficient multi-variate polynomials $R[x_1,\dots,x_m]$ is also an integral domain. 
\end{lemma}

\subsection{Schwartz-Zippel Lemma}

\begin{lemma}[\cite{DL78,Zippel79,Schwartz80}]\label{lmm: Schwartz-Zippel}
Let $R$ be an integral domain and $S\subseteq R$ be a finite subset of $R$. For any $m,d\in\N$ and any non-zero $m$-variate polynomial $f:R^m\to\R$ of total degree $d$
\[
\Pr_{\vec x\in S^m}[f(\vec x) = 0] \le \frac{d}{|S|},
\]
where $\vec x=(x_1,\dots,x_m)$ is uniformly sampled from $S^m$. 
\end{lemma}

\subsection{Transcendental Number Theory}

\begin{definition}
A complex number $\alpha$ is called \emph{algebraic} if it is the root of a non-zero integer-coefficient polynomial of finite degree, and called \emph{transcendental} if it is not algebraic.
\end{definition}

In particular, any rational number $\alpha=p/q$ is algebraic as it is the root of the degree-$1$ integer-coefficient polynomial $qx-p$.   

\begin{theorem}[Lindemann–Weierstrass Theorem {\cite[Theorem 1.4]{Baker-book}}]\label{thm: LW}
    If $\alpha_1,\dots,\alpha_n$ are distinct algebraic numbers, $e^{\alpha_1},\dots,e^{\alpha_n}$ are linearly independent over algebraic numbers. In particular, $e$ is transcendental.
\end{theorem}

\section{Circuit and Query Models}
\label{sec: models}

Now we formally define our circuit and query models for the identity testing problem. 

\subsection{Definition of the Circuit Model}
\label{sec:def_circuit_model}

\newcommand{\C}{\mathsf{C}}
\newcommand{\Frac}{\mathsf{Frac}}
\newcommand{\Exp}{\mathsf{Exp}}

The circuit model we introduce next extends the standard \emph{arithmetic circuit model} by \emph{exponentiation gates} that is intended to model the exponential function $\exp(x)\eqdef e^x$ over real numbers. Formally, an $\AExp^1$ circuit is a DAG consisting of vertices that are either an input variable or a gate of the following types (all gates are of fan-out $1$): 
\begin{compactenum}
    \item constant gates $\C_a$ of fan-in $0$ intended to model a fixed integer $a$; 
    \item addition gates $\Sigma(x_1,\dots,x_m)$ of unbounded fan-in that are intended to model the addition of real numbers, i.e., $x_1+\dots+x_m$; 
    \item multiplication gates $\Pi(x_1,\dots,x_m)$ of unbounded fan-in that are intended to model the multiplication of real numbers, i.e., $x_1x_2\dots x_n$. 
    \item division gates $\Frac(x,y)$ of fan-in $2$ intended to model the division of real numbers, i.e., $x/y$; 
    \item exponentiation gates $\Exp(x)$ of fan-in $1$ intended to model the exponential function, i.e., $e^x$. 
\end{compactenum}
In addition, we do not allow compositions of exponentiation gates in $\AExp^1$ circuits; that is, for any $\AExp^1$ circuit, there is at most \emph{one} exponentiation gate along any path from an input variable to an output gate.\footnote{Similarly, one can define $\AExp^k$ circuits where there are at most $k$ exponential gates along any such path. In this work, we focus on $\AExp^1$ circuits as it is natural and is more relevant to the practical motivation of this work.} We assume by default that an $\AExp^1$ circuit has only one output gate, though the definitions and our results can be easily generalized to multi-output circuits.  

\paragraph{Evaluation of $\AExp^1$ Circuits over Real Numbers.} The evaluation of $\AExp^1$ circuits over real numbers should be clear through the definition. Let $C$ be any $n$-input $\AExp^1$ circuit and $\vec x\in\R^n$, the evaluation of $C(\vec x)$ is defined as the output of the output gate, where the output values of gates are defined by gate-by-gate evaluation following a topological order. 

In particular, if the divisor of a division gate is zero, $C(\vec x)$ is undefined. Therefore, an $\AExp^1$ circuit may compute a partial function $C:\R^n\to\R$. We define the \emph{domain} of a $C(\vec x)$ circuit on $\R$, denoted by $\dom_\R(C)$, as the set of $\vec x\in\R^n$ such that $C(\vec x)$ is defined. For any $\vec x\notin\dom_\R(C)$, we may also say $C(\vec x)=\bot$.

\paragraph{Practical Motivation: Modeling Components in Neural Networks.} Our circuit model is a straightforward formalization of the program representation in Mirage \cite{Mirage24}. Intuitively, the motivation to introduce $\AExp^1$ circuit is to model components that are widely used in model artificial neural networks that consist of non-linear activation functions. For instance, the softmax function 
\[ 
\softmax(\vec x) \eqdef \frac{\left(e^{x_1},\dots,e^{x_n}\right) }{\sum_{i=1}^n e^{x_i}} 
\] 
can be formalized as an $n$-output depth-$3$ $\AExp^1$ circuit. As a slightly more complicated example, the attention function \cite{attention} that is widely used in large language models (LLMs):
\[ 
\mathsf{Attention}(Q, K, V) = \softmax\left(\frac{QK^\top}{\sqrt{d_k}}\right)V,
\]
where $d_k$ is a constant and $Q$, $K$, $V$ are matrices, can be implemented by a depth-$8$ $\AExp^1$ circuit. To see this, notice that matrix multiplication can be implemented by a multi-output depth-$2$ arithmetic circuit (i.e.~a summation of multiplications), $\softmax$ can be implemented as above in depth $3$, and the division (by scalar) can be implemented in depth $1$.\footnote{Here, we assume for simplicity that $\sqrt{d_k}$ is an integer. If it is not the case, one can also estimate $\sqrt{d_k}$ by a rational number $p/q$, which can be implemented by a depth-$2$ circuit.}

\subsection{Query Model over Reals, and Why it is not Satisfactory}

Since $\AExp^1$ circuits compute functions over $\R$, a natural idealized model for the identity testing problem is that the algorithm can evaluate the circuit over any real input. Formally, a identity testing algorithm for $\AExp^1$ circuits $C$ in \emph{real query model} is allowed to query the following oracle: 
\begin{compactitem}
\item (\emph{Evaluation}). Given an input $\vec x\in\Q^n$, the evaluation oracle reports whether $C(\vec x)$ is undefined, $C(\vec x) = 0$, or $C(\vec x) \in \R\setminus\{0\}$.  
\end{compactitem}

The main problem of the real query model is that computers cannot deal with real numbers. For real-world applications, real numbers are usually approximated by \emph{floating-point numbers}, which does not satisfy arithmetic laws such as the commutativity and associativity of addition. This leads \emph{two-sided errors} in the implementation of a testing algorithm in real query model: 
\begin{compactitem}
\item (\emph{Completeness}). Even if $C$ is identically zero over $\R$, it may evaluate to a non-zero value on some inputs when the evaluation queries are implemented in float-point numbers.  
\item (\emph{Soundness}). Even if $C$ is not identically zero, it may evaluate to $0$ on any input when evaluation queries are implemented in float-point numbers, either because of the precision issue or because of the failure of arithmetic laws. 
\end{compactitem}
The occurrences of errors on both sides reduce the consistency and reliability of the testing algorithms. Take the example of optimization compilers PET \cite{PET21} or Mirage \cite{Mirage24} that aim to detect redundancy in neural networks modeled as $\AExp^1$ circuits. Two sided error will make the algorithms less predictable and reliable to users; indeed, it could be possible that the oracle will hardly ever report $C(\vec x)\in\R\setminus\{0\}$ even if $C(\vec x)$ is defined and non-zero for most or all $\vec x$ due to the floating-point issue. In that case, the optimization compiler will barely find any redundancy. It is worth noting that the error is \emph{on top of} the completeness and soundness error of the identity testing algorithm, so it cannot be resolved by a better (e.g., deterministic) testing algorithm.

\subsection{Query Model over Finite Fields}
\label{sec: query over field}

Next, we introduce a natural query model for the identity testing problem of $\AExp^1$ circuits over finite fields that can be implemented in real-world applications \emph{without} the precision issue.

\paragraph{Evaluation over Finite Fields.} Let $p,q$ be two primes such that $q\mid p-1$, and $G_{p,q}$ be the (unique) multiplicative subgroup of $\F_p^*$ of order $q$. Equivalently, $G_{p,q}$ contains the roots of the univariate polynomial $z^q-1$ over $\F_p$. Loosely speaking, we will define the evaluation of $C$ on the input $(\vec u,\vec v)\in\F_p^n\times \F_q^n$ and $a\in G_{p,q}$ (modulo $(p,q)$) by 
\begin{compactitem}
\item interpreting the exponentiation gate by the function $x\mapsto a^x\bmod p$,
\item implementing all computation ``on the exponent'' in $\F_q$, 
\item implementing all computation ``under the exponent'' in $\F_p$, 
\end{compactitem}
and evaluating the circuit gate by gate. 

Specifically, given an $\AExp^1$ circuit $C$, $\vec u=(u_1,\dots,u_n)\in\F_p^n$, $\vec v=(v_1,\dots,v_n)\in\F_q^n$, and $a\in G_{p,q}$, we can implement the evaluation $C_a(\vec u,\vec v)$ as follows. For every $i\in[n]$, we assign a pair $(u_i,v_i)\in(\F_p\cup\{\perp\})\times(\F_q\cup\{\perp\})$ to the $i$-th input variable $x_i$ and evaluate gate by gate following the rules below: 
\begin{compactenum}
\item $\C_b$ is assigned $(b\bmod p, b\bmod q)$;
\item $\Sigma((\tau_1,\alpha_1),\dots,(\tau_m,\alpha_m))$ is assigned $(\sum_{i\in[m]}\tau_i\bmod p,\sum_{i\in[m]}\alpha_i\bmod q)$;
\item $\Pi((\tau_1,\alpha_1),\dots,(\tau_m,\alpha_m))$ is assigned $(\tau_1\tau_2\dots\tau_m\bmod p,\alpha_1\alpha_2\dots\alpha_m\bmod q)$.
\item $\Frac((\tau,\alpha),(\tau',\alpha'))$ is assigned $(\tau\cdot \Inv_p(\tau')\bmod p,\alpha\cdot \Inv_q(\alpha')\bmod q)$ if $\tau'\ne 0$ and $\alpha'\ne 0$, where $\Inv_p(u)$ denotes the multiplicative inverse of $u$ modulo $p$. If $\tau'=0$ and/or $\alpha'=0$, the first and/or the second coordinate are replaced by $\bot$. 
\item $\Exp((\tau,\alpha))$ is assigned $(a^{\alpha}\bmod p,\perp)$. 
\end{compactenum}
Note that summation, division, and multiplication denote the corresponding arithmetic operations in $\F_p\cup\{\bot\}$ (in the first coordinate) and $\F_q\cup\{\bot\}$ (in the second coordinate). In particular, for the evaluation of the exponentiation gate, $a^\alpha$ is defined by identifying $\alpha\in\F_q$ as an integer in $\Z_p=\{0,1,\dots,q-1\}$. 

Similar to the evaluation of $\AExp^1$ circuits on $\R$, $C_a(\vec u,\vec v)$ may be undefined, denoted by $C_a(\vec u,\vec v) = \bot$, due to division-by-zero. For each $a\in G_{p,q}$, we define the domain of $C$ with respect to $a$ modulo $(p,q)$, denoted by $\dom_{p,q,a}(C)$, as $\{(\vec u,\vec v)\in\F_p^n\times \F_q^n\mid C_a(\vec u,\vec v)\ne \bot\}$. 

\paragraph{Algebraic Query Model.} Subsequently, we define the identity testing problem for an $\AExp^1$ circuit $C$ over the algebraic query model as follows. Let $p,q,G_{p,q}$ be specified above. The algorithm for identity testing is allowed to query the following oracle: 
\begin{compactitem}
\item (\emph{Evaluation}). Given $(\vec u,\vec v)\in\F_p^n\times \F_q^n$ and $a\in G_{p,q}$, the oracle reports $C_a(\vec u,\vec v)\in\F_p\cup\{\bot\}$. 
\end{compactitem}

This query model is efficient both theoretically and practically. It can be observed, for instance, that the oracle in the algebraic query model can be implemented by algorithms in $O(s\cdot (\log p+\log q))$ space and polynomial time. In practice, the simplicity of the evaluation model makes it possible to implement it in the neural network compiler scenario (see \cite{PET21,Mirage24} for more details).  

\section{Exponential Polynomials}
\label{sec: exp polynomials}

To understand the functions computed by $\AExp^1$ circuits, we will define \emph{exponential polynomials} that, intuitively, generalize (multi-variate) polynomials by allowing terms of form $\exp(\cdot)$. Formally: 

\begin{definition}[Exponential polynomial]\label{def: exp polynomials}
    Let $R$ be an integral domain and $d,m$ be integers, $d\ge 0$, $m\ge 1$. An $m$-variate $R$-coefficient exponential polynomial $P(\vec x)$ of degree $d$ is defined as 
    \[
        P(\vec x) \eqdef f_1(\vec x) \exp\left(\frac{g_1(\vec x)}{h_1(\vec x)}\right) + f_2(\vec x) \exp\left(\frac{g_2(\vec x)}{h_2(\vec x)}\right) + \dots + f_k(\vec x) \exp\left(\frac{g_k(\vec x)}{h_k(\vec x)}\right), 
    \]
    where $\vec x=(x_1,\dots,x_m)$ denotes the indeterminates, $f_i$, $g_i$ and $h_i$ are $m$-variate polynomials of degree $d$ with coefficients from $R$. 
    
    The number of terms $k$ is said to be the \emph{width} of $P(\vec x)$, $\{f_i(\vec x)\}_{i \in [k]}$ are said to be the \emph{coefficient polynomials} of $P(\vec x)$, and $\{g_i(\vec x)/h_i(\vec x)\}_{i \in [k]}$ are said to be the \emph{exponent fractions} of $P(\vec x)$. 
\end{definition} 

In this paper, we only use the special case where $R=\Z$, i.e., integer-coefficient exponential polynomials. Nevertheless, we will develop the elementary arithmetic of exponential polynomials with respect to an arbitrary integral domain $R$ for coefficients. 

\subsection{Basic Arithmetic Properties}

We stress that an exponential polynomial should be considered as an abstract expression rather than a function. In particular, $\exp(\cdot)$ should be considered as a symbol instead of the exponential function over $\R$. For simplicity, we will also use the summation symbol $\sum$ to define an exponential polynomial, i.e., 
\[
P(\vec x) \eqdef \sum_{i=1}^{k} f_i(\vec x) \exp\left(\frac{g_i(\vec x)}{h_i(\vec x)}\right),
\]
where the summation symbol is a shorthand of the $k$-term summation in \Cref{def: exp polynomials}. 

We first define the addition and multiplication of exponential polynomials. Let $R$ be a ring and $P(\vec x), P'(\vec x)$ be $R$-coefficient exponential polynomials defined as 
\begin{align}
    P(\vec x) \eqdef \sum_{i=1}^{k} f_i(\vec x) \exp\left(\frac{g_i(\vec x)}{h_i(\vec x)}\right),  \quad P'(\vec x) \eqdef \sum_{i=1}^{k'} f'_i(\vec x) \exp\left(\frac{g'_i(\vec x)}{h'_i(\vec x)}\right).  
\end{align}
We can naturally define the addition and multiplication of $P(\vec x)$ and $P'(\vec x)$ as: 
\begin{align}
& P(\vec x) + P'(\vec x) \eqdef \sum_{i=1}^{k} f_i(\vec x) \exp\left(\frac{g_i(\vec x)}{h_i(\vec x)}\right) + \sum_{i=1}^{k'} f'_i(\vec x) \exp\left(\frac{g'_i(\vec x)}{h'_i(\vec x)}\right). \\ 
& P(\vec x) \cdot P'(\vec x) \eqdef \sum_{i=1}^k\sum_{j=1}^{k'} f_i(\vec x)f'_j(\vec x) \exp\left(\frac{g_i(\vec x)h'_j(\vec x) + g'_j(\vec x)h_i(\vec x)}{h_i(\vec x)h'_j(\vec x)}\right).
\end{align}

Next, we consider the arithmetic laws for exponential polynomials. As hinted at in the definitions of addition and multiplication, both operations are \emph{commutative} and \emph{associative}, and multiplication is \emph{distributive} over addition. Moreover, it is implicit in the definition of multiplication that exponentiation symbol satisfies
\begin{equation}
\exp\left(\frac{g(\vec x)}{h(\vec x)}\right) \cdot \exp\left(\frac{g'(\vec x)}{h'(\vec x)}\right) = \exp\left(\frac{g(\vec x)h'(\vec x) + g'(\vec x) h(\vec x)}{h(\vec x)h'(\vec x)}\right). 
\end{equation}
Furthermore, we impose the following axiom that allows merging terms with the same exponent fraction: 
\begin{align}
& f(\vec x) \exp\left(\frac{g(\vec x)}{h(\vec x)}\right) + f'(\vec x) \exp\left(\frac{g(\vec x)}{h(\vec x)}\right) = (f(\vec x) + f'(\vec x))\exp\left(\frac{g(\vec x)}{h(\vec x)}\right).
\end{align}
We say that $P(\vec x)$ and $P'(\vec x)$ are \emph{identical}, denoted by $P(\vec x) = P'(\vec x)$, if they can be transformed to each other using the arithmetic laws above. 

We note that the following two exponential polynomials 
\[
P_1(\vec x) = \exp\left(x\right);\quad P_2(\vec x) = \exp\left(\frac{x^2-2x}{x-2}\right)
\]
are not identical as the division law is \emph{not} allowed in the exponentiation symbol. This is intentional as in the evaluation of $\AExp^1$ circuits, we will only evaluate the circuit gate by gate without trying to simplify the circuit using the division law.

\subsection{Evaluation of Exponential Polynomials}

Similar to standard multivariate polynomials, we can define the evaluation of exponential polynomials that explains how to view such abstract expressions as functions. In particular, we will introduce two definitions corresponding to the \emph{real evaluation model} and \emph{algebraic evaluation model} of $\AExp^1$ circuits. 

Let $P(\vec x)$ be an $n$-variate integer-coefficient exponential polynomial defined by the coefficient polynomials $\{f_i\}_{i\in [k]}$ and exponent fractions $\{g_i/h_i\}_{i\in[k]}$. 

\paragraph{Evaluation of Exponential Polynomials on Real Numbers.} We can view $P(\vec x)$ as a function $P(\cdot):\R^n\to\R$ as follows. Given an input $\vec{u} \in \R^n$, the evaluation of $P(\vec x)$ on the input $\vec{u}$ is defined as 
\[ 
P(\vec{u}) \eqdef \sum_{i=1}^k f_i(\vec u) \cdot \exp\left(\frac{g_i(\vec u)}{h_i(\vec u)}\right)
\] 
where all the operators (i.e., additions, divisions, multiplications, and exponentiations) are interpreted as corresponding functions in $\R$. If for some $i\in [k]$, $h_i(\vec u) = 0$, the exponential polynomial is said to be undefined on $\vec u$. The domain of $P(\vec x)$ over $\R$, denoted by $\dom_\R(P)$, is defined as $\dom_\R(P)\eqdef \{\vec u\in\R^n\mid h_i(\vec u)\ne 0\; \forall i\in[k]\}$. 

The following proposition immediately follows from the definition and arithmetic laws.  

\begin{proposition}
Let $P(\vec x)$ and $P'(\vec x)$ be identical integer-coefficient exponential polynomials. Then $\dom_\R(P) = \dom_\R(P')$ and $P(\vec u) = P'(\vec u)$ for any $\vec u\in\R^n$. 
\end{proposition}

\paragraph{Evaluation of Exponential Polynomials on Finite Fields.} Similar to the algebraic query model for $\AExp^1$ circuits, we require two finite fields for the computation under and over exponents to define the evaluation of $P$ over finite fields. Let $p, q$ be two primes such that $q \mid p - 1$, and $G_{p, q}$ be the multiplicative subgroup of $\F_p^*$ of order $q$. Given inputs $\vec{u} \in \F_p^n, \vec{v} \in \F_q^n$, and $a \in G_{p, q}$, the evaluation of $P$ on $\vec{u}$ and $\vec{u}$, denoted by $P_a(\vec{u}, \vec{v})$, is defined as: 
\[
P_a(\vec u,\vec v) \eqdef \left(\sum_{i=1}^k f_i(\vec u) \cdot a^{g_i(\vec v)\cdot (h_i(\vec v))^{-1}\bmod q}\right) \bmod p.
\]
Note that $P_a(\vec u, \vec v)$ is undefined, denoted by $P_a(\vec u,\vec v)=\bot$, if $h_i(\vec v) \bmod q = 0$ for some $i \in [k]$. The domain of $P(\vec x)$ over $(p,q,a)$ is defined as 
\[ 
\dom_{p,q,a}(P(\vec x))\eqdef \{(\vec u,\vec v)\in\F_p^n\times \F_q^n \mid h_i(\vec u)\ne 0\bmod q\;\forall i\in[k]\}.
\] 

\begin{proposition}
Let $P(\vec x)$ and $P'(\vec x)$ be identical integer-coefficient exponential polynomials, $p,q$ be two primes such that $q\mid p-1$. For any $a$ in the multiplicative subgroup of $\F_p^*$ of order $q$, we have that $\dom_{p,q,a}(P) = \dom_{p,q,a}(P')$ and $P_a(\vec u,\vec v) = P'_a(\vec u,\vec v)$ for any $\vec u\in\F_p^n,\vec v\in\F_q^n$.  
\end{proposition}

\subsection{Condensation of Exponential Polynomials}
\label{sec: condensation}

Recall that two exponential polynomials are said to be identical if they can be transformed to each other by arithmetic laws. As there is no division law for the exponent fractions, the exponential polynomials 
\[
P_1(x)\eqdef \exp\left(x\right); \quad P_2(x)\eqdef \exp\left(\frac{x^2-2x}{x-2}\right)
\]
are not considered to be the same polynomial. In particular, we can notice that $\dom_\R(P_1) \ne \dom_\R(P_2)$. Nevertheless, these two exponential polynomials as functions are essentially the same over all but the input $x=2$. 

We now introduce the \emph{condensation} of integer-coefficient exponential polynomials that simplifies an exponential polynomial by merging terms with ``essentially the same'' exponent fractions together, which will be useful in proving our main results. 

\paragraph{Equivalent Exponential Fractions.} Let $P(\vec x)$ be a multi-variate integer-coefficient exponential polynomial of width $k$, $\{g_i(\vec x)/h_i(\vec x)\}_{i \in [k]}$ be the exponent fractions of $P(\vec x)$. Suppose that for each $i\in[k]$, $h_i$ is not a zero polynomial. We define $\sim_P$ to be the relation over $[k]$ such that 
\[
i \sim_P j \quad \text{ iff } \quad g_i(\vec x) h_j(\vec x) - g_j(\vec x)h_i(\vec x) = 0.
\]
Note that $g_i(\vec x) h_j(\vec x) - g_j(\vec x)h_i(\vec x) = 0$ means that it is a zero integer-coefficient polynomial, or equivalently, $g_i(\vec u) h_j(\vec u) - g_j(\vec u)h_i(\vec u) = 0$ for any $\vec x\in\vec \R$. 

\begin{lemma}
    Suppose that $h_i(\vec x)$ is not a zero polynomial for any $i\in[k]$, then $\sim_P$ is an equivalence relation over $[k]$. 
\end{lemma}

\begin{proof}
    Clearly $\sim_P$ is reflexive and symmetric, so it suffices to prove that $\sim_P$ is transitive. Let $i_1,i_2,i_3\in [k]$ such that $i_1\sim_P i_2$ and $i_2\sim_P i_3$, we will prove that $i_1\sim_P i_3$. Without loss of generality, we assume that $i_1=1,i_2=2,i_3=3$. 
    
    By the assumption, we have that
    \[
        g_1(\vec x)h_2(\vec x) = g_2(\vec x)h_1(\vec x), \quad g_2(\vec x)h_3(\vec x) = g_3(\vec x)h_2(\vec x),
    \]
    which implies that 
    \[
        h_2(\vec x)\left(g_1(\vec x)h_3(\vec x) - g_3(\vec x)h_1(\vec x)\right) = g_2(\vec x)h_1(\vec x)h_3(\vec x) - g_2(\vec x)h_3(\vec x)h_1(\vec x) = 0.
    \]
    As $\Z[\vec x]$ is an integral domain (see \Cref{lmm: polynomial integral domain}) and $h_2(\vec x)$ is not a zero polynomial, we have that $g_1(\vec x)h_3(\vec x) - g_3(\vec x)h_1(\vec x) = 0$, i.e., $i_1 \sim_p i_3$.
\end{proof}

\paragraph{Condensation of Exponential Polynomials.} Subsequently, we define a condensation $\hat P$ of an exponential polynomial $P$ as obtained by grouping the coefficient polynomials according to the relation $\sim_P$. Formally: 

\begin{definition}[Condensation of exponential polynomials]
    Let $P(\vec x)$ be an integer-coefficient exponential polynomial of degree $d$ and width $k$, $\{f_i(\vec x)\}_{i \in [k]}$ be the coefficient polynomials of $P(\vec x)$, $\{g_i(\vec x) / h_i(\vec x)\}$ be the exponent fractions of $P(\vec x)$ such that $h_i(\vec x)$ is not a zero polynomial for each $i\in[k]$. Let $\pi=\{[i_1]_\pi,[i_2]_\pi,\dots,[i_t]_\pi\}$ be the partition of $[k]$ induced by $\sim_P$, and $i_1,\dots,i_t$ be arbitrary representation elements. We say that $\hat P(\vec x)$ is a condensation of $P(\vec x)$ if it is of form 
    \[
    \tilde P(\vec x) \eqdef F_1(\vec x)\cdot \exp\left(\frac{g_{i_1}(\vec x)}{h_{i_1}(\vec x)}\right) + \dots + F_t(\vec x)\cdot \exp\left(\frac{g_{i_t}(\vec x)}{h_{i_t}(\vec x)}\right) 
    \]
    where $F_j(\vec x)\eqdef \sum_{i\in [i_j]_\pi} f_i(\vec x)$ is an integer-coefficient polynomial of degree at most $\hat d$.
\end{definition}

\begin{definition}[Condensed exponential polynomials]
    Let $P(\vec x)$ be an exponential polynomial with exponent fractions $\{g_i(\vec x)/h_i(\vec x)\}_{i \in [k]}$. $P(\vec x)$ is a condensed exponential polynomial if $h_i$'s are not zero polynomials and $g_i(\vec x)h_j(\vec x) \neq h_i(\vec x)g_j(\vec x)$ for $i, j \in [k]$ and $i \neq j$.
\end{definition}

We stress that the condensation of an exponential polynomial is not unique as we can choose the representative elements $i_1,\dots,i_t$ arbitrarily from their equivalent classes. The following proposition shows that $\hat P$ may have a larger domain compared to $P$, but they agree on the domain of $\hat P$. 

\begin{proposition}\label{prop: condensation}
Let $\hat P$ be a condensation of an integer-coefficient exponential polynomial $P$. Then:
\begin{compactitem}
\item $\dom_\R(P)\subseteq\dom_\R(\hat P)$, and $P,\hat P$ agree on $\dom_\R(P)$. 
\item Let $p,q$ be prime numbers such that $q\mid p-1$, $G_{p,q}$ be the multiplicative subgroup of $\F_p^*$ of order $q$, and $a\in G_{p,q}$. Then $\dom_{p,q,a}(P)\subseteq \dom_{p,q,a}(\hat P)$, and $P,\hat P$ agree on $\dom_{p,q,a}(P)$. 
\end{compactitem}
\end{proposition}

\begin{proof}[Proof Sketch]
We only prove the first bullet; the proof for the second bullet is similar. It is clear that $\dom_\R(P)\subseteq\dom_\R(\hat P)$ as each exponent fraction of $\hat P$ is an exponential fraction of $P$, so it suffices to prove that $P,\hat P$ agree on $\dom_\R(\hat P)$. 

Let $\vec u\in \dom_\R(\hat P)$. We know by definition that $h_i(\vec u)\ne 0$ for every $i\in[k]$. In that case, for any $j_1,j_2\in [k]$ such that $j_1\sim_P j_2$, we know that the $j_1$-th and the $j_2$-th exponential fractions evaluate to the same value, which subsequently implies that $P(\vec u) = \hat P(\vec u)$. 
\end{proof}

\subsection{Structural Lemma for \texorpdfstring{$\AExp^1$}{AExp1} Circuits}

Now we are ready to prove a structural lemma showing that $\AExp^1$ circuits can be seen as fractions of exponential polynomials. 

For simplicity, we introduce the following notation. Let $P,P'$ be $n$-variate exponential polynomials, we define $\dom_\R(P/P')$ be the set $\{\vec u\in \dom_\R(P)\cap\dom_\R(P')\mid P'(\vec u)\ne 0\}$, i.e., the domain of the fraction $P(\vec x)/P'(\vec x)$. Similarly, let $p,q$ be prime numbers such that $q\mid p-1$, $G_{p,q}$ be the multiplicative subgroup of $\F_p^*$ of order $q$, and $a\in G_{p,q}$, we define $\dom_{p,q,a}(P/P')$ as the set $\{\vec u\in\dom_{p,q,a}(P)\cap \dom_{p,q,a}(P')\mid P'(\vec u) \ne 0\}$. 

\begin{restatable}{lemma}{NormalLemma}\label{lmm: normal form}
    For every $n$-input $\AExp^1$ circuit $C$, there are $n$-variate integer-coefficient exponential polynomials $P(\vec x)$ and $P'(\vec x)$ such that the following holds: 
    \begin{compactitem}
    \item $\dom_\R(C) = \dom_\R(P/P')$, and for each $\vec u\in\dom_\R(C)$, $C(\vec u) = P(\vec u)/P'(\vec u)$. 
    \item Let $p,q$ be prime numbers such that $q\mid p-1$, $G_{p,q}$ be the multiplicative subgroup of $\F_p^*$ of order $q$, and $a\in G_{p,q}$. Then $\dom_{p,q,a}(C)=\dom_{p,q,a}(P/P')$ and for every $(\vec u,\vec v)\in\dom_{p,q,a}(C)$, $C_a(\vec u,\vec v)\equiv P_a(\vec u,\vec v)\cdot\Inv_p(P'_a(\vec u,\vec v))\pmod p$. 
    \end{compactitem}
    In particular, if $C$ does not contain exponentiation gates, both $P$ and $P'$ do not contain the exponentiation symbol, i.e., $P$ and $P'$ are integer-coefficient polynomials. 
\end{restatable}

\begin{proof}
    We prove the lemma by induction on the depth $d$ of $C$. Note that the base case is trivial as an $\AExp^1$ circuit of depth $0$ can only output one of its input $x_i$ and can thus be converted into the fraction of $P(\vec x)\eqdef x_i$ and $P'(\vec x) \eqdef 1$. 

    In the induction case, we know by the induction hypothesis the the lemma holds for circuits of depth at most $d-1$. We consider the type of the output gate $G_o$ of $C$. 

    (\emph{Constant}). Suppose the output gate $G_o$ is a constant gate $\C_a$ always outputs $a$. Similarly to the base case, $C(\vec x)$ can be written as the fraction of $P(\vec x) \eqdef a$ and $P'(\vec x) \eqdef 1$. It is clear that neither of $P$ and $P'$ contains the exponential symbol.

    (\emph{Addition}). Suppose that the output gate $G_o$ is an addition gate whose input wires are connected to the gates $G_1,G_2,\dots,G_\ell$. Let $C_1,C_2,\dots,C_\ell$ be the sub-circuits of $C$ with $G_1,G_2,\dots,G_\ell$ being the output gates, respectively. By the induction hypothesis, we know that for each $i\in[\ell]$, there are exponential polynomials $P_i(\vec x),P'_i(\vec x)$ that satisfy the lemma for $C_i$. Let $P(\vec x),P'(\vec x)$ be the following exponential polynomials
    \begin{align*}
    P(\vec x) &\eqdef \sum_{i\in[\ell]} P_i(\vec x) \cdot \prod_{j\in[\ell]\setminus\{i\}} P'_j(\vec x) \\ 
    P'(\vec x) &\eqdef \prod_{i\in[\ell]} P_i'(\vec x)
    \end{align*}
    The two bullets in the lemma can be verified straightforwardly, and we will only verify the first bullet. It is clear that both $\dom_\R(C)$ and $\dom_\R(P/P')$ are the intersection of $\dom_\R(P_i),\dom_\R(P'_i)$, and $D_i\eqdef \{\vec x\in\R^n\mid P'_i(\vec x)\ne 0\}$ for $i\in[\ell]$; for each $\vec x\in\dom_\R(P/P')$, we know by the definition of evaluation of $C(\vec x)$, $P(\vec x)$, and $P'(\vec x)$ that 
    \[
    C(\vec x) = \sum_{i\in[\ell]} C_i(\vec x) = \sum_{i\in[\ell]} \frac{P_i(\vec x)}{P'_i(\vec x)} = \frac{P(\vec x)}{P'(\vec x)}, 
    \]
    where the second equality follows from the induction hypothesis, and the last equality follows from $\vec x\in \dom_\R(P/P')\subseteq D_i$ for each $i\in[\ell]$.

    If $C$ does not contain exponentiation gates, then none of $\{C_i\}_{i \in \ell}$ contains exponentiation gates, as they are sub-circuits of $C$. Since the construction of $P$ and $P'$ does not introduce the exponentiation symbol, by the induction hypothesis, neither $P$ nor $P'$ contains it.

    (\emph{Multiplication}). Similarly to the addition case, suppose that the output gate $G_o$ is a multiplication gate whose input wires are connected to the gates $G_1, G_2, \ldots, G_l$. Let $C_1, C_2, \ldots, C_\ell$ be the sub-circuits of $C$ with $G_1, G_2, \ldots, G_\ell$ being the output gates, respectively. By the induction hypothesis, we know that for each $i \in [\ell]$, there are exponential polynomials $P_i(\vec x)$, $P'_i(\vec x)$ that satisfy the lemma for $C_i$. Let $P(\vec x)$, $P'_i(\vec x)$ be the following exponential polynomials
    \begin{align*}
        P(\vec x) &\eqdef \prod_{i \in [\ell]} P_i(\vec x), \\ 
        P'(\vec x) &\eqdef \prod_{i \in [\ell]} P'_i(\vec x).
    \end{align*}
    We only verify the first bullet. It is clear that both $\dom_\R(C)$ and $\dom_\R(P/P')$ are the intersection of $\dom_\R(P_i),\dom_\R(P'_i)$, and $D_i\eqdef \{\vec x\in\R^n\mid P'_i(\vec x)\ne 0\}$ for $i\in[\ell]$; for each $\vec x\in\dom_\R(P/P')$, we know by the definition of evaluation of $C(\vec x)$, $P(\vec x)$, and $P'(\vec x)$ that 
    \[
    C(\vec x) = \prod_{i\in[\ell]} C_i(\vec x) = \prod_{i\in[\ell]} \frac{P_i(\vec x)}{P'_i(\vec x)} = \frac{\prod_{i\in[\ell]} P_i(\vec x)}{\prod_{i\in[\ell]} P'_i(\vec x)} = \frac{P(\vec x)}{P'(\vec x)}, 
    \]
    where the second equality follows from the induction hypothesis.

    Through a similar argument to in the addition case, $P$ and $P'$ do not contain any exponentiation symbol if $C$ does not contain any exponentiation gate. 
    
    (\emph{Division}). Suppose that the output gate $G_o$ is a division gate whose input wires are connected to the gates $G_1$ and $G_2$. Let $C_1$ and $C_2$ be the sub-circuites of $C$ with $G_1$ and $G_2$ being the output gates, respectively. By the induction hypothesis, we know that there are polynomials $P_1(\vec x)$, $P'_1(\vec x)$, $P_2(\vec x)$, $P'_2(\vec x)$ that satisfy the lemma for $C_1$ and $C_2$, respectively. Let $P(\vec x)$, $P'(\vec x)$ be the following exponential polynomials
    \begin{align*}
        P(\vec x) &\eqdef P_1(\vec x) \cdot P'_2(\vec x), \\
        P'(\vec x) &\eqdef P_2(\vec x) \cdot P'_1(\vec x).
    \end{align*}
    We only verify the first bullet. Both $\dom_\R(C)$ and $\dom_\R(P/P')$ are the intersection of $\dom_\R(P_i)$, $\dom_\R(P'_i)$, $D \eqdef \{\vec x \in \R^n \mid P'_i(\vec x)\neq 0\}$ for $i \in [2]$, and $D_i \eqdef \{\vec x \in \R^n \mid P_i(\vec x)\neq 0\}$; for each $\vec x \in \dom_\R(P/P')$, we know by the definition of evaluation of $C(\vec x), P(\vec x)$, and $P'(\vec x)$ that
    \[
    C(\vec x) = \frac{C_1(\vec x)}{C_2(\vec x)} = \frac{P_1(\vec x) / P'_1(\vec x)}{P_2(\vec x) / P'_2(\vec x)} = \frac{P_1(\vec x) \cdot P'_2(\vec x)}{P_2(\vec x) \cdot P'_1(\vec x)} = \frac{P(\vec x)}{P'(\vec x)},
    \]
    where the second equality follows from the induction hypothesis, and the third equality follows from $\vec x \in \dom_\R(C) \subseteq D_2$.

    Similarly, both $P$ and $P'$ do not contain any exponentiation symbol if $C$ does not contain any exponentiation gate. 

    (\emph{Exponentiation}). Suppose that the output gate $G_o$ is an exponentiation gate whose input wire is connected to the gate $C_1$ with an output gate $G_1$. Note that in $\AExp^1$ circuits, there is at most one exponentiation gate along each path, so $C_1$ does not contain any exponentiation gates. By the induction hypothesis, there are integer-coefficient multi-variate polynomials $P_1$ and $P'_1$ satisfying the lemma for $C_1$. We can write $C$ as the fraction of $P(\vec x) = \exp(P_1/P'_1)$ and $P'(\vec x) = 1$.
\end{proof}

\subsection{Width, Degree, and Weight of Concrete Neural Network Components}
\label{sec: real world NN}
It is proved in \Cref{lmm: normal form} that any $\AExp^1$ circuit can be transformed to an equivalent fraction of exponential polynomials. We can therefore define the \emph{width}, \emph{degree}, and \emph{weight} of an $\AExp^1$ circuit. 

\begin{definition}[width, degree, and weight of $\AExp^1$ circuits]
    An $\AExp^1$ circuit $C$ is said to have width $k$, degree $d$, and weight $w$ if there are integer-coefficient degree-$d$ width-$k$ exponential polynomials $P(\vec x)$ and $P'(\vec x)$ that satisfy both conditions in \Cref{lmm: normal form}.
\end{definition}

We note that the transformation in \Cref{lmm: normal form} may not be efficient --- the width, degree, and the bit-length of the integer coefficients may grow exponentially with respect to the size of the circuit. Nevertheless, it can be verified that many neural network components can be simulated by fractions of exponential polynomials with relatively small width, degree, and coefficients. 

\begin{example}[$\softmax$ function]
$\softmax$ is a widely used activation function in neural networks. Given as input a matrix $M \in \R^{m \times n}$, it is defined as
\[
\softmax(M)_{ij} \eqdef \frac{\exp(M_{ij})}{\sum_{k=1}^n \exp(M_{ik})}.
\]
By the definition, we can see that for each $i \in [m], j \in [n]$, $\softmax(M)_{ij}$ can be computed by an $\AExp^1$ circuit of width $n$, degree $1$ and weight $1$.
\end{example}

\begin{example}[Attention]
Attention \cite{attention} is a widely used building block in large language models defined as 
\[
\Attention(Q, K, V) \eqdef  \softmax\left(\frac{QK^\top}{\sqrt{d_k}}\right)V
\]
where $Q \in \R^{m \times d_k} , K \in \R^{n \times d_k}, V \in \R^{n \times d_v}$ are input matrices. Assume for simplicity that $\sqrt{d_k}\in\N$. Let $A = \frac{QK^\top}{\sqrt{d_k}}\in\R^{m\times n} $. For each $i \in [m], j \in [n]$, we have that 
\begin{align*}
\Attention(Q, K, V)_{ij} &= \sum_{k=1}^n \softmax(A)_{ik} V_{kj} \\
&= \sum_{k=1}^n \frac{\exp(A_{ik})V_{kj}}{\sum_{l=1}^n \exp(A_{il})} \\
&= \frac{\sum_{k=1}^n V_{kj}\exp(A_{ik})}{\sum_{l=1}^n \exp(A_{il})}.
\end{align*}
Each entry of $A$ can be computed by a fraction of two integer-coefficient polynomials of degree $2$ and weight $\sqrt{d_k}$: 
\[
A_{ik} = \frac{\sum_{l=1}^{d_k} Q_{il}\cdot K_{lk}}{\sqrt{d_k}}.
\]
Subsequently, $\Attention(Q, K, V)$ can be computed by an $\AExp^1$ circuit of width $n$, degree $2$ and weight $\sqrt{d_k}$.
\end{example}

As shown in the two examples above, one could expect that many neural network components can be formalized as $n$-input $\AExp^1$ circuits with constant degree, $\poly(n)$ weight, and $\poly(n)$ width. It is worth noting that some other neural network components may be even more efficient in terms of weight and width: 

\begin{example}
Besides $\softmax(\cdot)$, other activation functions are used for neural networks in general, and many of them can be efficiently implemented as $\AExp^1$ circuits. For instance, the GLU function \cite{glu} is defined as
\[
\GLU(\vec x, W, V) \eqdef \sigmoid(\vec x\cdot W) \otimes (\vec x\cdot V),
\]
where $x \in \R^n, W, V \in \R^{n \times m}$ are inputs, $\sigmoid(\vec z)_i \eqdef \frac{1}{1+\exp(-z_i)}$, and $\otimes$ denotes coordinate-wise multiplication. For every $i\in[m]$,  
\[
\GLU(\vec x,W,V)_i = \frac{(\vec x\cdot V)_i}{1+\exp(-(\vec x\cdot W)_i)},
\]
and thus it can be computed by an $\AExp^1$ circuit of width $2$, degree $2$ and weight $1$.
\end{example}

\section{Algorithms for Identity Testing}
\label{sec: algorithms}

Now we are ready to describe our identity testing algorithms in real and algebraic query models. Indeed, our algorithms are essentially the same one: Randomly sample an input $x$ (in corresponding models) and check whether the circuit evaluates to zero or $\bot$ on the input $\vec x$. We will first prove the correctness of the algorithm in real query model (see \Cref{sec: correctness real model}), and generalize the proof to the algebraic query models in subsequent subsections.  

\subsection{Identity Testing in Real Query Model}
\label{sec: correctness real model}

Formally, the identity testing algorithm in real query model works as follows: Suppose that $C:\R^n\to\R$ is an $\AExp^1$ circuit of width $k$ and degree $d$, let $B=20dk^2$ be sufficiently large and $S = \{1,2,\dots,B\}$, the algorithm uniformly sample $x_1,x_2,\dots,x_n\in S$, and accept if $C(x_1,\dots,x_n) \ne 0$. 

It is clear that the algorithm is perfectly complete, and the soundness can be formalized as the following theorem. 

\begin{theorem}\label{thm: soundness}
    Let $C:\R^n\to\R$ be an $\AExp^1$ circuit of width $k$ and degree $d$ that is not identically zero on $\dom_\R(C)$. Then for any non-empty finite subset $S\subseteq\Q$, if $\vec x\in S^n$ is sampled uniformly at random, $\Pr[C(\vec x) \in\{0,\bot\} ] \le 8dk^2/|S|$. 
\end{theorem}

The key step for proving \Cref{thm: soundness} is a necessary and sufficient condition for an exponential polynomial $P$ to be identically zero. Intuitively, if the exponent fractions $g_1/h_1,\dots,g_k/h_k$ of $P$ are pairwisely distinct, then $P$ is identically zero on $\R^n$ if and only if all coefficient polynomials $f_1,\dots,f_k$ are all identically zero. Formally: 

\begin{lemma}\label{lmm: equiv pexp}
    Let $P:\R^n\to\R$ be a condensed exponential polynomial of form  
    \[
    P(\vec x) = \sum_{i = 1}^k f_i(\vec x)\cdot \exp\left(\frac{g_i(\vec x)}{h_i(\vec x)}\right),
    \]
    where $f_i,g_i,h_i$ are integer-coefficient polynomials of total degree $d$ and $h_i(\vec x)$ is not identically zero on $\R^n$ for every $i\in[k]$. Then the following statements are equivalent: 
    \begin{compactenum} 
    \item $P$ is identically zero on $\dom_\R(P)$; 
    \item $\Pr[P(\vec x) \in \{0,\bot\}] > 3dk^2/|S|$ for $\vec x\in S^n$ sampled uniformly at random for any non-empty finite subset $S\subseteq\Q$. 
    \item For every $i\in[k]$, $f_i$ is identically zero on $\R^n$. 
    \end{compactenum}
\end{lemma}

\newcommand{\sE}{\mathcal{E}}
    
\begin{proof}
    Note that $(1)\Rightarrow (2)$ and $(3)\Rightarrow (1)$ are trivial, so it suffices to prove $(2)\Rightarrow (3)$. Let $S\subseteq \Q$ be a non-empty set. Let $\zeta\subseteq[k]$ be the set of indices such that $f_i$ is identically zero on $\R^n$. Suppose that $\zeta\ne\varnothing$, we need to prove that with probability at least $1-3dk^2/|S|$ on $S^n$, $P(\vec x)$ is defined and $P(\vec x) \ne 0$.
    
    Let $\vec x$ be a random variable sampled uniformly from $S^n$, we calculate the probability of the following events:
    \begin{compactitem}
    \item Let $\sE_h$ be the event that $h_i(\vec x) \ne 0$ for every $i\in\zeta$. Note that for every fixed $i\in\zeta$, we know by \Cref{lmm: Schwartz-Zippel} that $\Pr[h_i(\vec x)= 0]\le d / |S|$, and therefore by the union bound, $\Pr[\sE_h]=1-\Pr[\lnot\sE_h]\ge 1-dk/|S|$. 
    \item Let $\sE_{g}$ be the event that for every $i,j\in\zeta$, $i\ne j$, $g_i(\vec x)\cdot h_j(\vec x) \ne g_j(\vec x)\cdot h_i(\vec x)$. Note that for each pair $i,j\in\zeta$, $g_i(\vec x)\cdot h_j(\vec x)-g_j(\vec x)\cdot h_i(\vec x)$ is a non-zero integer-coefficient $n$-variate polynomials of total degree at most $2d$, thus by \Cref{lmm: Schwartz-Zippel}, we have that  
    \[
    \Pr[g_i(\vec x)\cdot h_j(\vec x)-g_j(\vec x)\cdot h_i(\vec x) = 0] \le \frac{2d}{|S|}. 
    \]
    Then by the union bound, we know that $\Pr[\sE_{g}] = 1-\Pr[\lnot\sE_g] \ge 1- dk^2/|S|$. 
    \item Let $\sE_f$ be the event that for every $i\in\zeta$, $f_i(\vec x) \ne 0$. Similar to $\sE_h$, we know that $\Pr[\sE_f]\ge 1-dk/|S|$.
    \end{compactitem}
    
    Subsequently, we know that $\Pr[\sE_h\land \sE_g\land \sE_f] \ge 1-3dk^2/|S|$. We now argue that $P(\vec x)$ is defined and $P(\vec x) \ne 0$ for any $\vec x\in S^n$ such that $\sE_h,\sE_g,\sE_f$ are true simultaneously. Fix any such $\vec x$. For every $i\in\zeta$, define $\beta_i \eqdef f_i(\vec x)$ and $\alpha_i \eqdef g_i(\vec x) / h_i(\vec x)$ (note that $h_i(\vec x)\ne 0$ by $\sE_h$), we know that:
    \begin{compactitem} 
    \item For each $i\in\zeta$, $\beta_i,\alpha_i\in\Q$, as $\vec x\in S^n\subseteq \Q^n$ and $\Q$ is close under addition and multiplication. 
    \item For each $i\in\zeta$, $\beta_i \ne 0$ by $\sE_f$. 
    \item For each $i,j\in\zeta$, $i\ne j$, $\alpha_i\ne \alpha_j$ by $\sE_g$.
    \end{compactitem}
    
    Since $P(\vec x) = \beta_1 e^{\alpha_1} + \dots + \beta_k e^{\alpha_k} = \sum_{i\in\zeta} \beta_i e^{\alpha_i}$ and $\zeta\ne\varnothing$, we know by the Lindemann-Weierstrass Theorem (see \Cref{thm: LW}) that $P(\vec x) \ne 0$, which completes the proof.
\end{proof}

We are now ready to prove \Cref{thm: soundness}. 

\begin{proof}[Proof of \Cref{thm: soundness}]
    Let $C:\R^n\to\R$ be an $\AExp^1$ circuit of width $k$ and degree $d$ that is not identically zero on $\dom_\R(C)$. Let $S\subseteq\Q$ be a non-empty finite set. Then by the definition of degree and width of $C$, there are exponential polynomials $P,P'$ defined by 
    \begin{align*}
    P(\vec x) &\eqdef \sum_{i\in[k]} f_i(\vec x)\cdot\exp\left(\frac{g_i(\vec x)}{h_i(\vec x)}\right),\\ 
    P'(\vec x) &\eqdef \sum_{j\in[k]} f_j'(\vec x)\cdot\exp\left(\frac{g_j'(\vec x)}{h_j'(\vec x)}\right).
    \end{align*}
    such that $C(\vec x) = P(\vec x)/P'(\vec x)$ for every $\vec x\in\R^n$, and for each $i\in[k],j\in[k]$, $f_i,g_i,h_i,f'_j,g'_j,h'_j$ are integer-coefficient $n$-variate polynomials of total degree at most $d$.

    Note that for each $i\in[k], j\in[k]$, we have that $h_i(\vec x),h_j'(\vec x)$ are not identically zero on $\R^n$. This is because otherwise at least one of $P(\vec x),P'(\vec x)$ is undefined on every $\vec x\in \R^n$, which implies that $\dom_\R(C) = \varnothing$. Moreover, we know that $P$ and $P'$ are not identically zero on their domains, respectively, as otherwise $C$ must be identically zero on its domain. 

    Let $\hat P(\vec x)$ be a condensation of $P(\vec x)$. As $P$ is not identically zero on its domain, we know by \Cref{prop: condensation} that $\hat P(\vec x)$ is not identically zero on its domain. Consider the random variable $\vec x\in S^n$ sampled uniformly at random. We calculate the probability of the following events: 
    \begin{compactitem}
    \item Let $\sE_{\hat P}$ be the event that $\hat P(\vec x)\notin\{0,\bot\}$, i.e., $\hat P(\vec x)$ is defined and $\hat P(\vec x)\ne 0$. By \Cref{lmm: equiv pexp}, we can see that $\Pr[\sE_{\hat P}] \ge 1-3dk^2/|S|$. 
    \item Let $\sE_{\bot}$ be the event that $P(\vec x)$ is defined, i.e., $h_i(\vec x) \ne 0$ for every $i\in[k]$. Note that $\Pr[h_i(\vec x) = 0] \le d/|S|$ for every fixed $i\in[k]$ by \Cref{lmm: Schwartz-Zippel}. By the union bound, we know that $\Pr[\sE_\bot] = 1-\Pr[\lnot \sE_\bot] \ge 1-dk/|S|$. 
    \item Let $\sE_{P}$ be the event that $P(\vec x)\notin\{0,\bot\}$. Notice that 
    \[
    \Pr[\sE_{\hat P}] = \Pr[\sE_{\hat P}\land \sE_\bot] + \Pr[\sE_{\hat P}\land \lnot \sE_\bot] \ge 1-3dk^2/|S|, 
    \]
    where $\Pr[\sE_{\hat P}\land \sE_\bot] = \Pr[\sE_P]$ by \Cref{prop: condensation} and $\Pr[\sE_{\hat P}\land \lnot \sE_\bot] \le \Pr[\lnot \sE_\bot]\le dk/|S|$. Therefore, $\Pr[\sE_P] \ge 1-3dk^2/|S|-dk/|S|\ge 1-4dk^2/|S|$. 
    \end{compactitem}
    
    Following the same argument, we can prove that $\Pr[P'(\vec x)\notin\{0,\bot\}] \ge 1-4dk^2/|S|$ over uniformly random $\vec x\in S^n$. It follows that 
    \begin{align*}
    \Pr[C(\vec x)\in\{0,\bot\}] & = \Pr[P(\vec x)\in\{0,\bot\}\lor P'(\vec x)\in\{0,\bot\}] \\ 
    & \le \Pr[P(\vec x)\in\{0,\bot\}] + \Pr[P'(\vec x)\in\{0,\bot\}] \le 8dk^2/|S|. 
    \end{align*}
    This completes the proof. 
\end{proof}

\subsection{A Weak Descartes' Rule over Finite Fields}

\begin{theorem}\label{thm: week lw}
    Let $p$, $q$ be prime numbers such that $q\mid p-1$, $g\in\F_p$ be an element of order $q$, and $G\eqdef\langle g \rangle$ be the unique multiplicative subgroup of $\F_p^*$ of order $q$. Let $k\le q$, $\alpha_1,\dots,\alpha_k\in \F_q$ be distinct, and $\beta_1,\dots,\beta_k\in\F_p$. Then the univariate polynomial
    \[
    f(z)\eqdef \beta_1 z^{\alpha_1} + \beta_2 z^{\alpha_2} + \dots + \beta_k z^{\alpha_k}
    \]
    has at most $q^{1-1/(k-1)}$ roots in $G$.
\end{theorem}

The key idea of the proof (which follows from \cite{Canetti00,Kelley16}) is to find a low-degree polynomial that has as many roots as $f(z)$ in $G$. Let $p,q$ be prime numbers such that $q\mid p-1$, and $G=\langle g\rangle$ be the unique multiplicative subgroup of $\F^*_p$ of order $q$. Formally: 

\begin{proposition}\label{prop: week lw inner}
    Let $c \in \{1,2,\dots,q-1\}$ and $f_c(z)$ be the polynomial
    \[
        f_c(z) \eqdef \beta_1 z^{c\alpha_1\bmod q} + \beta_2 z^{c\alpha_2\bmod q} + \dots + \beta_k z^{c\alpha_k\bmod q}.
    \]
    Then the polynomials $\{f_c(z)\}_{c \in [q-1]}$ have the same number of roots in $G$.
    \end{proposition}
    \begin{proof}
        Let $c \in [q-1]$ and $c > 1$. We will prove that $f_c(z)$ and $f_1(z)=f(z)$ have the same number of roots in $G$. We first prove that the number of roots of $f(z)$ in $G$ is smaller than that of $f_c(z)$. For any $g^a$ that is a root of $f(z)$ in $G$, let $c^{-1}$ be the multiplicative inverse of $u$ modulo $q$, then
        \begin{align*}
            0 &= \beta_1 g^{a\alpha_1\bmod q} + \beta_2 g^{a\alpha_2\bmod q} + \dots + \beta_k g^{a\alpha_k\bmod q} \\
            &= \beta_1 g^{ac^{-1}c\alpha_1\bmod q} + \beta_2 g^{ac^{-1}c\alpha_2\bmod q} + \dots + \beta_k g^{ac^{-1}c\alpha_k\bmod q}. 
        \end{align*}
        This implies that $g^{ac^{-1}}$ is a root of $f_c(z)$. As the mapping $g^{a}\mapsto g^{a c^{-1}\bmod q}$ (as a function over $G$) is injective, the number of roots of $f(z)$ in $G$ is smaller than that of $f_c(z)$. 

        It remains to prove that the number of roots of $f_c(z)$ in $G$ is smaller than that of $f(z)$. This is because for any root $g^a$ of $f_c(z)$ in $G$, $g^{ac\bmod q}$ is a root of $f(z)$, and the mapping $g^{a}\mapsto g^{ac\bmod q}$ is injective. This concludes the proof. 
    \end{proof}

    We will need the following lemma from \cite{Kelley16}: 

    \begin{lemma}[Lemma 4.1 of \cite{Kelley16}]\label{lmm: Kelley}
    For $\alpha_1,\dots,\alpha_k,N\in\N$ and $n\le N/\gcd(\alpha_1,\dots,\alpha_t,N)$, there is a $c\in[n-1]$ such that 
    \[
    \max_{i\in[t]}\{\alpha_i\cdot c\bmod N\}\le \frac{N}{n^{1/t}}. 
    \]
    \end{lemma}

    \begin{proof}[Proof of \Cref{thm: week lw}]
    Fix $p,q\in\N,g\in\F_p,k$, elements $\alpha_1,\dots,\alpha_k\in\F_q$ and $\beta_1,\dots,\beta_k\in\F_p$. Let $f_c(z)\in\F_p[z]$ be the univariate polynomial 
    \[
    f_c(z) \eqdef \beta_1 z^{c\alpha_1\bmod q} + \beta_2 z^{c\alpha_2\bmod q} + \dots + \beta_k z^{c\alpha_k\bmod q}
    \]
    and $f(z) = f_1(z)$. Assume that $\alpha_1\le \alpha_2\le \dots \le \alpha_k$. We may further assume without loss of generality that $\alpha_1=0$, as otherwise we can consider the polynomial $f(z)/z^{\alpha_1}$ that has the same number of roots in $G$. 
    
    Let $N\eqdef q$ and $n=q$, we have $\gcd(\alpha_2,\dots,\alpha_k,N)=1$ and thus $n\le N/\gcd(\alpha_2,\dots,\alpha_k,N)$. By \Cref{lmm: Kelley}, there exists some $c\in\{1,2,\dots,q-1\}$ such that for every $i\in\{2,3,\dots,k\}$, 
    \[
    \alpha_i\cdot c\bmod q\le q^{1-1/(k-1)},
    \] 
    which further implies that the polynomial $f_c(z)$ is of degree at most $q^{1-1/(k-1)}$. Subsequently, $f_c(z)$ has at most $q^{1-1/(k-1)}$ roots in $G$. This completes the proof as $f(z)$ and $f_c(z)$ have the same number of roots in $G$ by \Cref{prop: week lw inner}.
    \end{proof}

\subsection{Identity Testing in Algebraic Query Model}

\MainTheorem*

We first prove two lemmas: \Cref{lmm: for completeness} is used to prove the completeness property, and \Cref{lmm: zero poly ff} is used to prove the soundness property. 

\begin{lemma}\label{lmm: for completeness}
    Let $\hat P$ be a condensation of an integer-coefficient exponential polynomial $P$. If $P$ is identically zero on $\dom_\R(P)$ and $\dom_\R(P)\ne \varnothing$, then $\hat P$ is identically zero on $\dom_\R(\hat P)$. 
\end{lemma}

\begin{proof}
    Fix $k,d$ such that both $P$ and $\hat P$ are of degree at most $d$ and width at most $k$. Suppose, towards a contradiction, that $\hat P$ is not identically zero on $\dom_\R(\hat P)$. Let $S\subseteq\Z$ be a set of size at least $10 dk^2$. By \Cref{lmm: equiv pexp}, we know that for a uniformly random $\vec x\in S^n$, $\Pr[\hat P(\vec x)=0]\le 1/3$.

    Let $g_1/h_1,\dots,g_k/h_k$ be the exponential fractions of $P$. We know that none of $h_1,\dots,h_k$ is a zero polynomial in $\R$ as otherwise $\dom_\R(P)=\varnothing$. Then for each $i\in[k]$, if we sample $\vec x\in S^n$ uniformly at random, $\Pr[h_i(\vec x) = 0] \le d/|S|$ by \Cref{lmm: Schwartz-Zippel}. Subsequently, by the union bound, we have that with probability at most $dk/|S|\le 1/10$, $h_i(\vec x)\ne 0$ for some $i\in[k]$, or equivalently, $\vec x\notin\dom_\R(P)$. 

    Combining these two cases and the union bound, for $x\in S^n$ sampled uniformly at random, 
    \[
    \Pr[\vec x\notin\dom_\R(P)\lor \hat P(\vec x)= 0] \le \Pr[\hat P(\vec x)=0] + \Pr[\vec x\notin \dom_\R(P)] <1. 
    \]
    Therefore, there is an $x\in S^n\subseteq\Z^n$ such that $\vec x\in\dom_\R(P)$ and $\hat P(\vec x) \ne 0$. In such case, $P(\vec x) \ne 0$ by \Cref{prop: condensation}, which leads to a contradiction. 
\end{proof}

\begin{lemma}\label{lmm: zero poly ff}
Let $p, q$ be prime numbers such that $q \mid  p - 1$, $G_{p, q}$ be the unique multiplicative subgroup of $\F_p^*$ of order $q$. Let $P$ be a condensed integer-coefficient exponential polynomial 
\[
P(\vec x) = \sum_{i\in[k]} f_i(\vec x) \cdot \exp\left(\frac{g_i(\vec x)}{h_i(\vec x)}\right)
\]
of degree $d$ and width $k$. Suppose that for every $i\in[k]$, $f_i(\vec x)$ is not identically zero on $\F_p$ and $h_i(\vec x)$ is not identically zero on $\F_q$, then for $(\vec u, \vec v, a)$ uniformly sampled from $\F_p^n\times\F_q^n\times G_{p, q}$,
\[ 
\Pr[P_a(\vec u, \vec v)  \in \{0, \bot\}] \le 3dk^2\cdot q^{-1} + q^{-1/(k-1)}. 
\] 
\end{lemma}

\begin{proof}
Fix $p,q$ and let $\{f_i\}_{i \in [k]}$ and $\{g_i/h_i\}_{i\in [k]}$ be the coefficient polynomials and exponent fractions of $P$, respectively. We know from the assumption that $f_i\ne 0\pmod p$ and $h_i\ne 0\pmod q$. 

Let $(\vec u, \vec v,a)$ be random variables sampled uniformly from $\F_p^n\times\F_q^n\times G_{p, q}$. We calculate the probability of the following events:
\begin{compactitem}
\item Let $\sE_h$ be the event that $h_i(\vec v) \neq 0$ for every $i \in [k]$. By Lemma~\ref{lmm: Schwartz-Zippel}, for a fixed $i \in [k]$, $\Pr[h_i(\vec v) = 0] \le d/q$. Then by the union bound, $\Pr[\sE_h] = 1 - \Pr[\neg \sE_h] \ge 1 - dk/q$.
\item Let $\sE_g$ be the event that for every $i \neq j$ from $[k]$, $g_i(\vec v)h_j(\vec v) \neq g_j(\vec v)h_i(\vec v)$. Since $g_i(\vec x)h_j(\vec x) - g_j(\vec x)h_i(\vec x)$ is a non-zero integer-coefficient polynomial of total degree at most $2d$, by Lemma~\ref{lmm: Schwartz-Zippel}, $\Pr[g_i(\vec v)h_j(\vec v) - g_j(\vec v)h_i(\vec v) = 0] \le 2d/q$. Then by the union bound, $\Pr[\sE_g] \ge 1 - dk^2/q$.
\item Let $\sE_f$ be the event that for every $i \in [k], f_i(\vec x) \neq 0$. Similar to $\sE_h$, $\Pr[\sE_f] \ge 1 - dk/q$.
\end{compactitem}
Let $\beta_i=f_i(\vec u)\in\F_p$ and $\alpha_i=g_i(\vec v)\cdot \Inv_q(h_i(\vec v))\in\F_q$ for every $i\in[k]$. When $\sE_h, \sE_g$ and $\sE_f$ are all true, we have that 
\[
P_a(\vec u,\vec v) = \beta_1 a^{\alpha_1} + \dots + \beta_k a^{\alpha_k}, 
\]
where $\beta_1,\dots\beta_k\ne 0$ and $\alpha_1,\dots,\alpha_k$ are distinct. By Lemma~\ref{thm: week lw}, $\Pr[P_a(\vec u, \vec v) \bmod{p} \in \{0, \bot\} \mid \sE_h \wedge \sE_g \wedge \sE_h] \le q^{-1/(k-1)}$. Therefore, 
\begin{align*}
    & \Pr[P_a(\vec u, \vec v) \in \{0, \bot\}] \\ 
 \le ~& \Pr[P_a(\vec u, \vec v) \notin \{0, \bot\} \mid \sE_h \wedge\sE_g \wedge \sE_h] + \Pr[\lnot(\sE_h \wedge \sE_g \wedge \sE_f)] \\ 
 \le~& 3dk^2\cdot q^{-1} + q^{-1/(k-1)}.
\end{align*}
This completes the proof. 
\end{proof}

\begin{proof}[Proof of \Cref{thm: algebraic model}]
By the definition of width, degree of $\AExp$ circuits, there are exponential polynomials
\[P(\vec x) = \sum_{i\in[k]} f_i(\vec x) \cdot \exp\left(\frac{g_i(\vec x)}{h_i(\vec x)}\right), P'(\vec x) = \sum_{i\in[k]} f'_i(\vec x) \cdot \exp\left(\frac{g'_i(\vec x)}{h'_i(\vec x)}\right)
\]
such that the following hold:
\begin{compactitem}
\item $f_i,g_i,h_i,f'_i,g_i',h'_i$ are $d$-degree polynomials with coefficients in $[-w, w] \cap \Z$,
\item $\dom_\R(C) = \dom_\R(P/P')$, and $C(\vec u) = P(\vec u)/P'(\vec u)$ for $\vec u \in \dom_\R(C)$, 
\item $\dom_{p,q,a}(C)=\dom_{p,q,a}(P/P')$, and $C_a(\vec u, \vec v) = P(\vec u, \vec v) \cdot \Inv_p(P'(\vec u, \vec v))$ for any $a \in G_{p,q}$ and $(\vec u, \vec v) \in \dom_{p,q,a}(C)$. 
\end{compactitem}

\paragraph{Proof of Completeness.} Let $C$ be an $\AExp^1$ circuit that is identically zero on $\dom_\R(C)$. We first prove that either $P(\vec x)$ is identically zero on $\dom_\R(P)$, or $P'(\vec x)$ is identically zero on $\dom_\R(P')$.

Towards a contradiction, we assume that $P$ is not identically zero on $\dom_\R(P)$ and $P'$ is not identically zero on $\dom_\R(P')$. Let $S\subseteq\Z$ be a set of size $30dk^2$. By \Cref{thm: soundness}, we know that for $\vec x\in S^n$ sampled uniformly at random, 
\[
\Pr[P(\vec x) \in \{0,\bot\}], \Pr[P'(\vec x) \in \{0,\bot\}] < \frac{1}{3}. 
\]
By the union bound, we know that for $\vec x\in S^n$ sampled uniformly at random, with probability at least $1/3$, $P(\vec x),P'(\vec x)\notin\{0,\bot\}$. For any such $\vec x\in S^m$, we have that $\vec x\in\dom_\R(C)$ and $C(\vec x) = P(\vec x)/P'(\vec x)\ne 0$, which is impossible as $C$ is identically zero on its domain. 

Suppose that $P(\vec x)$ is identically zero on $\dom_\R(P)$, we know by \Cref{lmm: for completeness} that the condensation $\hat P(\vec x)$ of $P(\vec x)$ is also identically zero on $\dom_\R(\hat P)$. Let $k'$ be the width of $\hat P$ and 
\[
\hat P(\vec x) = \sum_{i\in[k']} \hat f_i(\vec x)\cdot \exp\left(\frac{\hat g_i(\vec x)}{\hat h_i(\vec x)}\right).
\]
By \Cref{lmm: equiv pexp}, we know that for every $i\in[k']$, $\hat f_i(\vec x)$ is identically zero on $\R^n$, or equivalently, $\hat f_i(\vec x)$ is a zero polynomial. In that case, we must have $\hat f_i(\vec x)\equiv 0\pmod p$ for every $p$. This implies that 
\[
\hat P_a(\vec u,\vec v)=\sum_{i\in[k']} \hat f_i(\vec u)\cdot \exp\left(\frac{\hat g_i(\vec v)}{\hat h_i(\vec v)}\right) \in\{0,\bot\}. 
\]
Subsequently, $P_a(\vec u,\vec v)\in\{0,\bot\}$ as $\hat P$ and $P$ agree on $\dom_{p,q,a}(P)$ (see \Cref{prop: condensation}), which further implies that $C_a(\vec u,\vec v)\in\{0,\bot\}$. 

Similarly, if $P'(\vec x)$ is identically zero on $\dom_\R(P')$, then $P'_a(\vec u,\vec v)\in\{0,\bot\}$, which implies that $C_a(\vec u,\vec v)=\bot$. This concludes the completeness of the algorithm.   

\paragraph{Proof of Soundness.} Let $R(\vec x) = P(\vec x) \cdot P'(\vec x)$. It can be verified that it is an exponential polynomial with width $k^2$, degree $2d$ and weight $w^2$. It follows that: 
\begin{compactitem}
    \item For any $\vec x \in \R^n$, $C(\vec x) = P(\vec x) / P'(\vec x) \in \{0, \bot\}$ if and only if $R(\vec x) \in \{0, \bot\}$.
    \item For any $\vec x \in \F_p^n \times \F_q^n, \vec a \in G_{p, q}$, $C_a(\vec x) = P_a(\vec x) / P'_a(\vec x) \in \{0, \bot\}$ if and only if $R_a(\vec x) \in \{0, \bot\}$.
\end{compactitem}

Since $C$ is not identically zero on $\dom_\R(C)$, i.e., $C(x)\notin \{0,\bot\}$ for some $\vec x\in\R^n$, $R$ is not identically zero on $\dom_\R(R)$. Let $\hat R$ be a condensation of $R$. By \Cref{prop: condensation}, $\hat R$ is not identically zero on $\dom_\R(\hat R)$.

Let $\{f''_i\}_{i \in [k^2]}$ be the coefficient polynomials of $\hat R$, and $\{g''_i/h''_i\}_{i \in [k^2]}$ be the exponent fractions of $\hat R$. Note that $h''_i$ is non-zero for every $i\in[k^2]$, as otherwise $\dom_{\R}(R)=\varnothing$. By \Cref{lmm: equiv pexp}, we know that $f''_i$ is non-zero for some $i\in[k^2]$ and, without loss of generality, we may assume that $f''_i$ is non-zero for every $i\in[k^2]$.

It can be verified that the integer weights in $h''_i$ are within $[-w^2,w^2]$. Moreover, the integer weights in $f''_i$ are within $[-(kw)^2,(kw)^2]$, as each $f''_i$ is a summation of at most $k^2$ polynomials that have integer weights within $[-w^2,w^2]$. As $p>q>2(kw)^2$, we know that for every $i\in[k^2]$, $f''_i$ is not identically zero on $\F_p^n$ and $h''_i$ is not identically zero on $\F_q^n$.

Let $\vec x = (\vec u, \vec v) \in \F_p^n \times \F_q^n$ and $a \in G_{p, q}$ be random variables sampled uniformly at random.  We calculate the probability of the following events:
\begin{compactitem}
\item Let $\sE_{\hat R}$ be the event that $\hat R(\vec x) \notin \{0, \bot\}$. By Lemma~\ref{lmm: zero poly ff}, $\Pr[\sE_{\hat R}] \ge 1 - 6dk^4/q - q^{-1/(k^2-1)}$.
\item Let $\sE_{\bot}$ be the event that $h''_i(\vec v)\neq 0$ for every $i \in [k^2]$. For every fixed $i \in [k^2]$, by Lemma~\ref{lmm: Schwartz-Zippel}, $\Pr[h''_i(\vec v) = 0] \le 2d/q$. Then by the union bound, $\Pr[\sE_{\bot}] = 1 - \Pr[\neg \sE_{\bot}] \ge 1 - 2dk^2/q$.
\item Let $\sE_R$ be the event that $R(\vec x) \notin \{0, \bot\}$. Notice that
\[
\Pr[\sE_{\hat R}] = \Pr[\sE_{\hat R} \wedge \sE_{\bot}] + \Pr[\sE_{\hat R} \wedge \neg\sE_{\bot}] \ge 1 - 6dk^4/q - q^{-1/(k^2-1)},
\]
where $\Pr[\sE_{\hat R} \wedge \sE_{\bot}] = \Pr[\sE_R]$ and $\Pr[\sE_{\hat R} \wedge \neg\sE_{\bot}] \le 2dk^2/q$. Therefore, 
\[ 
\Pr[\sE_R] \ge 1 - 6dk^4/q - q^{-1/(k^2-1)} - 2dk^2/q \ge 1 - 8dk^4/q - q^{-1/(k^2-1)}.
\] 
\end{compactitem}
This completes the proof, as $C_a(\vec x)\in\{0,\bot\}$ if and only if $R_a(\vec x)\in\{0,\bot\}$. 
\end{proof}

\bibliographystyle{alpha}
\bibliography{ref}

\end{document}